\documentclass[a4paper,12pt,openright]{report}

\usepackage{amsmath}
\usepackage{amssymb}
\usepackage{graphicx}
\usepackage{float}
\usepackage{caption}
\usepackage{amsthm}
\usepackage{color}
\usepackage[linesnumbered,algoruled,boxed,lined]{algorithm2e}

\newcommand{\rr}{\mathbb R}

\newcommand{\zz}{\mathbb Z}
\newcommand{\nn}{\mathbb N}
\newcommand{\qq}{\mathbb Q}

\newcommand{\ff}{\mathbb F}

\newcommand{\inv}{^{-1}}
\newcommand{\abs}[1]{\left|{#1}\right|}

\newcommand{\suchthat}{\ | \ }

\newcommand{\tth}{^\text{th}}
\newcommand{\genseq}[3]{{#1}_1 {#3} {#1}_2 {#3} \dots {#3} {#1}_{#2}}
\newcommand{\seq}[2]{\genseq{#1}{#2}{,}}

\newcommand{\twocases}[4]{\begin{cases} #2 & #1 \\ #4 & #3 \end{cases}}
\newcommand{\threecases}[6]{\begin{cases} #2 & #1 \\ #4 & #3 \\ #6 & #5 \end{cases}}

\newcommand{\txt}[1]{\text{#1}}

\newcommand{\stext}[1]{\ \ \ \ \ \text{(#1)}}
\newcommand{\stextn}[1]{\\&\ \ \ \ \ \ \stext{#1}}
\newcommand{\bcause}[1]{\stext{because ${#1}$}}
\newcommand{\snc}[1]{\stext{since ${#1}$}}

\newcommand{\push}{\\ & \ \ \ \ \ \ \ \ \ \ }

\newcommand{\ipnc}[3]{\begin{figure}[H]\begin{center}\includegraphics[scale = {#1}]{{#2}.pdf}\caption{#3}\end{center}\end{figure}}

\newcommand{\ipns}[4]{\begin{figure}[H]\begin{center}\includegraphics[scale = {#1}]{{#2}.pdf}\caption[#3]{#4}\end{center}\end{figure}}

\makeatletter 
\g@addto@macro{\@algocf@init}{\SetKwInOut{Parameter}{Parameters}} 
\makeatother

\newcommand{\PP}{\mathsf{P}}
\newcommand{\NP}{\mathsf{NP}}

\usepackage{epsfig,graphicx,verbatim,parskip,tabularx,setspace,xspace,enumitem,makecell,scrextend,mathtools,chngcntr}
\usepackage[numbers]{natbib}
\def\authorname{Jamie R.\ Tucker-Foltz\xspace}
\def\authorcollege{Churchill College\xspace}
\def\authoremail{jtuckerfoltz@gmail.com}
\def\dissertationtitle{Approximating Constraint Satisfaction Problems Symmetrically}

\theoremstyle{plain}
\newtheorem{theorem}{Theorem}
\newtheorem{lemma}[theorem]{Lemma} 
\newtheorem{proposition}[theorem]{Proposition}

\newtheorem{conjecture}[theorem]{Conjecture}

\numberwithin{theorem}{section}

\newcommand{\maxcut}{\textsf{MaxCut}\xspace}
\newcommand{\threecol}{\textsf{3-Colourability}\xspace}
\newcommand{\threesat}{\textsf{3SAT}\xspace}
\newcommand{\threexor}{\textsf{3XOR}\xspace}
\newcommand{\vc}{\textsf{VertexCover}\xspace}
\newcommand{\labelcover}{\textsf{LabelCover}\xspace}
\newcommand{\ug}{\textsf{UniqueGames}\xspace}
\newcommand{\tsp}{\textsf{TravelingSalesman}\xspace}
\newcommand{\UG}{\textsf{UG}}
\newcommand{\gug}{\textsf{GroupUniqueGames}\xspace}
\newcommand{\maxtwolin}{\textsf{Max2Lin}}

\newcommand{\gap}[2]{\ensuremath{\mathsf{Gap}_{{#1}, {#2}}}}

\newcommand{\thmspace}{\vspace{14pt}}
\newcommand{\LC}{\textsf{LC}\xspace}
\newcommand{\gaplambda}{\ensuremath{\textsf{Gap}_\Lambda}}
\newcommand{\opt}{\textup{opt}}
\newcommand{\sdp}{\textup{sdp}}
\newcommand{\tauuug}{\tau_{\textup{\UG}(q)}}
\newcommand{\weight}{\texttt{weight}}

\newcommand{\fm}{\ff_2^m}
\newcommand{\verteq}{\rotatebox{90}{$\,=$}}
\newcommand{\equalto}[2]{\underset{\scriptstyle\overset{\mkern4mu\verteq}{#2}}{#1}}

\makeatletter
\newcommand\fs@nocaptionruled{
	\let\@fs@capt\relax
	\def\@fs@pre{}
	\def\@fs@post{\kern2pt\hrule\relax}%
	\def\@fs@mid{\kern2pt\hrule\kern2pt}%
	\let\@fs@iftopcapt\iftrue}
\makeatother
\floatstyle{nocaptionruled}
\restylefloat{algorithm}

\newcommand{\lemGSameSatisfiability}{
	For any \gug instance $U$, the satisfiability of $\mathcal{G}(U)$ is the same as the satisfiability of $U$.
}
\newcommand{\lemUGExactClaim}{
	For all $i \geq 0$, for all $g \in A$, for all variables $x_{v_1}^{g_1}$ and $x_{v_2}^{g_2}$:
	\begin{enumerate}[label={(\arabic*)}]
		\item\label{itmUGExactGood2} If $\{v_1, v_2\} \neq r_i$, there is a constraint $x_{v_1}^{g_1} + x_{v_2}^{g_2} = g$ in $\mathcal{G}(U_1)$ if and only if there is a constraint $f_i(x_{v_1}^{g_1}) + f_i(x_{v_2}^{g_2}) = g$ in $\mathcal{G}(U_2)$.
		\item\label{itmUGExactGood3} If $\{v_1, v_2\} = r_i$, there is a constraint $x_{v_1}^{g_1} + x_{v_2}^{g_2} = g$ in $\mathcal{G}(U_1)$ if and only if there is \emph{not} a constraint $f_i(x_{v_1}^{g_1}) + f_i(x_{v_2}^{g_2}) = g$ in $\mathcal{G}(U_2)$.
	\end{enumerate}
}
\newcommand{\lemUGExactSoundness}{
	The satisfiability of $U_2$ (and thus of $\mathcal{G}(U_2)$) is strictly less than $\frac12$.
}
\newcommand{\lemNUGLGSoundnessUTwo}{
	With probability at least $\frac12 - \varepsilon$, the satisfiability of $U_2$ (and thus of $\mathcal{G}(U_2)$) is less than $\frac{\alpha}{2^{\ell}}$.
}
\newcommand{\lemNUGLGMostEdgesGood}{
	With probability at least $\frac12$, at most a $\gamma$ fraction of the edges of $\widetilde{H}$ are bad edges.
}
\newcommand{\lemNUGLGRadiusProperty}{
	Let $p = v_0, \seq{v}{n}$ be a path in $H$ of length $n \geq r$. Given any values in $\fm$ for $g^*(v_0)$ and $g^*(v_n)$, it is possible to extend $g^*$ to all of the intermediate vertices of $p$ so that the map $f(x_v^g) := x_v^{g + g^*(v)}$ is a partial isomorphism between $\mathcal{G}(U_1)$ and $\mathcal{G}(U_2)$ over the set $\{x_v^g \suchthat v \in p,\ g \in \fm\}$.
}
\newcommand{\lemGirth}{
	On any round $i$, for any vertex $u \in V(H)$, there does not exist any path contained in $F_i(u)$ with both endpoints in $T_{i - 1}$.
}

\begin{document}
 
\pagestyle{empty}
\singlespacing
\begin{titlepage} 
	
	\begin{center}
		\noindent
		\huge
		\dissertationtitle \\
		\vspace*{\stretch{1}}
	\end{center}
	
	\begin{center}
		\noindent
		\huge
		\authorname \\
		\Large
		\authorcollege      \\[24pt]
		\includegraphics{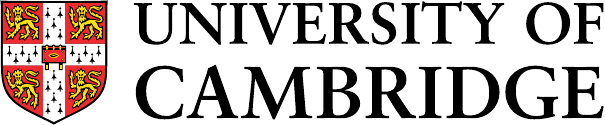}
	\end{center}
	
	\vspace{24pt} 
	
	\begin{center}
		\noindent
		\large
		{\it A dissertation submitted to the University of Cambridge \\ 
			in partial fulfilment of the requirements for the degree of \\ 
			Master of Philosophy in Advanced Computer Science} 
		\vspace*{\stretch{1}}
	\end{center}
	
	\begin{center}
		\noindent
		University of Cambridge \\
		Computer Laboratory     \\
		William Gates Building  \\
		15 JJ Thomson Avenue    \\
		Cambridge CB3 0FD       \\
		{\sc United Kingdom}    \\
	\end{center}
	
	\begin{center}
		\noindent
		Email: \authoremail \\
	\end{center}
	
	\begin{center}
		\noindent
		\today
	\end{center}
	
\end{titlepage} 

\newpage
\vspace*{\fill}

\onehalfspacing
\par\vspace*{.35\textheight}{\large\emph{Dedicated in memory of Lyle A. McGeoch.}}
\vspace*{\fill}

\singlespacing
\newpage
{\Huge \bf Abstract}
\vspace{24pt}

This thesis investigates the extent to which the optimal value of a constraint satisfaction problem (CSP) can be approximated by some sentence of fixed point logic with counting (FPC). It is known that, assuming $\PP \neq \NP$ and the Unique Games Conjecture, the best polynomial time approximation algorithm for any CSP is given by solving and rounding a specific semidefinite programming relaxation. We prove an analogue of this result for algorithms that are definable as FPC-interpretations, which holds without the assumption that $\PP \neq \NP$. While we are not able to drop (an FPC-version of) the Unique Games Conjecture as an assumption, we do present some partial results toward proving it. Specifically, we give a novel construction which shows that, for all $\alpha > 0$, there exists a positive integer $q = \txt{poly}(\frac{1}{\alpha})$ such that no there is no FPC-interpretation giving an $\alpha$-approximation of Unique Games on a label set of size $q$.

\newpage
\vspace*{\fill}

{\Huge \bf Acknowledgments}
\vspace{24pt} 

I would like to thank Anuj Dawar for suggesting this exciting and fruitful project for me to work on, and also for his unparalleled attentiveness and generosity as a supervisor. He has always been available to quickly answer my questions and read what I have written, from my vague outlines of ideas to my long and detailed proofs. I am extraordinarily grateful for all of the time and advice he has given me.

I would also like to thank the Winston Churchill Foundation for funding my year in Cambridge and generously supporting me when COVID-19 hit and I was suddenly forced to return to the USA.

Most importantly, a huge thank you to my parents who have kept me happy and well-fed as I wrote this the bulk of this thesis from home.

\vspace{200pt}

\pagenumbering{roman}
\setcounter{page}{0}
\pagestyle{plain}
\tableofcontents
\onehalfspacing



\chapter{Introduction}\label{chaIntro}
\pagenumbering{arabic}
\setcounter{page}{1} 

The study of \emph{approximation algorithms} asks the question, ``Given some $\NP$-hard optimization problem $\mathcal{P}$, what is the smallest multiplicative error to which we can approximate the optimal values of instances of $\mathcal{P}$ using a polynomial time algorithm?" For some problems, the error can be made arbitrarily small, while for others, there is a fundamental limit beyond which any better approximation could be used to solve the problem exactly, which is impossible unless $\PP = \NP$. In some cases it is known exactly what this limit is, while in others it is still open.

In this thesis we ask the same question, but with the additional requirement that the algorithm must be definable in \emph{fixed point logic with counting (FPC)}. Roughly, an algorithm is definable in FPC if it respects the natural symmetries of its input, without making any arbitrary choices that break those symmetries. (The formal connection between polynomial time algorithms, FPC and symmetry is discussed in greater detail in Section \ref{secFPCBackground}.)

In a recent paper, Atserias and Dawar \cite{DefinableInapproximabilityJournal} give the first (to the author's knowledge) FPC-\emph{inapproximability} results, showing that the problems \threexor, \threesat, \vc and \labelcover cannot be approximated closer than within specific constant factors in FPC. We extend this work to consider a broad class of problems, called \emph{constraint satisfaction problems (CSPs)}.

Using a reduction from the \ug problem to an arbitrary CSP $\Lambda$, Raghavendra \cite{RagThesis} shows that, assuming $\PP \neq \NP$ and the \emph{Unique Games Conjecture}, the best polynomial time approximation algorithm for $\Lambda$ is given by solving and rounding a specific semidefinite programming relaxation. In Chapter \ref{chaCSPs} we argue, firstly, that this algorithm is definable in FPC, and secondly, that the reduction from \ug is definable in FPC. Together, these two facts imply an analogue of Raghavendra's result for algorithms that are definable in FPC, holding without the assumption that $\PP \neq \NP$ (Theorems \ref{thmCSPApproxInFPC} and \ref{thmFPCUGHardness}).

However, the result still depends on an FPC-version of the Unique Games Conjecture (Conjecture \ref{cnjFPCUGC}). While we are not able to prove this conjecture, we do make some partial progress. In Chapter \ref{chaUG} we present a series of CFI-constructions\footnote{A CFI-construction is a construction used to prove a logical inexpressibility result similar to that of Cai, F\"urer and Immerman \cite{CaiFurerImmerman}. } culminating in the following result (Theorem \ref{thmUGLowGapMain}): for all $\alpha > 0$, there exists a positive integer $q = \txt{poly}(\frac{1}{\alpha})$ such that there is no FPC-definable $\alpha$-approximation algorithm for \ug on a label set of size $q$.


\chapter{Preliminaries}\label{chaBackground}

We assume that the reader is familiar with basic complexity theory, linear algebra, group theory and notation from graph theory. All graphs we consider are undirected, but may contain multiple edges between a pair of vertices and/or loops from a vertex to itself. A graph is \emph{simple} if it has no multiple edges or self-loops. We also assume an intuitive understanding of the meaning of sentences and formulas of first order logic.


\section{Constraint satisfaction problems and approximation}\label{secCSPsAndApproximationClassic}

An instance of a \emph{constraint satisfaction problem (CSP)} is specified by a set of variables taking values in some fixed, finite domain and a set of constraints between certain subsets of variables. The objective is to assign values to the variables to satisfy a maximum number of constraints, or, in a related \emph{weighted} version, to satisfy a set of constraints of maximum total weight.

The complexity of a CSP is determined by the size of the domain and the kinds of constraints which are allowed. For example, suppose the domain has size 3, and each constraint specifies that a certain pair of variables must not take the same value. Thinking of the constraints as edges in a graph, determining whether all constraints can be satisfied is the \threecol problem, which is $\NP$-complete. If, instead, the domain has size 2, then satisfying all constraints amounts to checking whether the graph is bipartite, which is in $\PP$. However, satisfying a \emph{maximum number} of constraints when it is impossible to satisfy all of them is still $\NP$-hard; this is the same as the \maxcut problem, where we have to partition the vertices of a graph into two sets such that a maximum number of edges are cut by the partition. We have a similar situation if the set of values is a finite field and the constraints are linear equations involving any number of variables: determining whether all constraints can be satisfied is solvable in polynomial time via Gaussian elimination, though satisfying the \emph{maximum number} of constraints when the system is inconsistent is $\NP$-hard.

Since finding the exact optimal value is $\NP$-hard for almost any interesting CSP, the next logical question is, is it possible to efficiently approximate the optimal value? For $0 \leq \alpha \leq 1$, an \emph{$\alpha$-approximation algorithm} for a CSP $\Lambda$ is a polynomial time algorithm that, given an instance of $\Lambda$ with optimal value $x^*$, returns a value $x$ such that $\alpha x^* \leq x \leq x^*$. Usually, an $\alpha$-approximation algorithm works by finding a specific assignment of variables $f$ and returning the number of constraints it satisfies; the hard part is in proving that there is no alternative assignment that beats $f$ by more than a factor of $\frac{1}{\alpha}$. The constant $\alpha$ is called the \emph{approximation ratio} attained by the algorithm.

For example, there is a greedy $\frac12$-approximation algorithm for \maxcut by Sahni and Gonzalez \cite{MAXCUTOneHalfApprox} which iteratively places vertices, in arbitrary order, on the side of the partition that maximizes the number of cut edges between the new vertex and the already-placed vertices. At each step, at least half of the new edges are cut, for otherwise the new vertex should have been placed on the other side. So by the end, if there are $m$ edges, at least $\frac{m}{2}$ edges are cut, i.e., the returned value of the cut $x$ must satisfy $\frac{m}{2} \leq x$. If $x^*$ denotes the maximum number of edges that can be cut, then $x^* \leq m$, so
$$\frac{x^*}{2} \leq \frac{m}{2} \leq x \leq x^*,$$
and thus we have a $\frac12$-approximation algorithm.

The analysis of the Sahni-Gonzalez algorithm reveals an important point about how one often thinks about approximately ``solving" a constraint satisfaction problem. This proof shows us that, instead of returning $x$, the algorithm could have instead just returned the value $\frac{E}{2}$. The greedy algorithm serves as nothing more than a proof of existence of such a cut, and actually carrying out this computation does not yield a better approximation ratio in the worst case. Intuitively, we expect that any reasonable algorithm for approximating a CSP should return not just the number of constraints satisfied, but also an assignment of values to variables which satisfies that many constraints. However, from a purely theoretical standpoint, this is unnecessary, and it is important to keep this in mind for Chapter \ref{chaCSPs} when we consider a setting where it is impossible to compute such an assignment.

\subsection{Semidefinite programming and the\\ Goemans-Williamson algorithm}\label{subSDPAndGW}
Semidefinite programming is an extremely powerful tool in the design of approximation algorithms. It is a generalization of linear programming which allows for certain kinds of nonlinear constraints, yet it is still solvable in polynomial time up to arbitrary precision. A \emph{semidefinite program (SDP)} is specified by an $n \times n$ objective matrix $C$, a collection of $n \times n$ constraint matrices $\{A_k\}_{k \in [m]}$ and a collection of corresponding constraint bounds $\{b_k\}_{k \in [m]}$. All numbers and matrices are $\qq$-valued. Given such an SDP, a \emph{feasible solution} is an $n \times n$ matrix $X$ such that, for all $k \in [m]$,
$$\langle A_k, X \rangle := \sum_{i \in [n]}\sum_{j \in [n]}A_{i, j}X_{i, j} \leq b_k,$$
and, additionally, $X$ is \emph{semidefinite}, written $X \succeq 0$. There are many equivalent definitions of semidefiniteness \cite[App.\txt{} A]{Aspects}; the most useful one for our purposes is that $X \succeq 0$ if there exists an $n \times n$ matrix $B$ such that $X = B^\top B$. An SDP is \emph{feasible} if the set of feasible solutions is nonempty, and \emph{bounded} if the set of feasible solutions is bounded. The objective is to find a feasible solution $X$ that maximizes the value of $\langle C, X \rangle$.

A paradigm in the design of approximation algorithms is to define a semidefinite program whose variables represent the variables of the input CSP instance, with the objective matrix capturing the quantity to be maximized in the problem. It is usually straightforward to write the constraints of the SDP so that an integral solution satisfying all of the constraints represents a valid solution to the CSP. However, after solving the SDP, we may get a solution with variables taking non-integral values, so the optimal SDP value may be larger than the value of the optimal integral solution. Thus, the final step is to ``round" the variables to integers, preserving feasibility without reducing the objective value too much. The approximation ratio attained by such an algorithm depends crucially on the analysis of the rounding step.

The canonical example of this technique is in the Goemans-Williamson algorithm for approximating \maxcut \cite{GWOriginal}. In the algorithm, a cut in an $n$-vertex graph is thought of as an assignment of $\pm 1$ to each of $n$ variables $\seq{v}{n}$, each representing one of the vertices, where the variables assigned 1 represent one side of the cut and the variables assigned $-1$ represent the other side. If an edge between vertex $i$ and vertex $j$ crosses the cut, then $v_iv_j = -1$, so $1 - v_iv_j = 2$. If such an edge does not cross the cut, we instead have $1 - v_iv_j = 0$. Therefore, the objective can be written as follows, where $w_{i, j}$ is the weight of the edge between vertex $i$ and vertex $j$ (so each $w_{i, j} \in \{0, 1\}$ for an unweighted graph):
\begin{align*}
	\textbf{Maximize }\ \ \ & \frac{1}{2}\sum_{i < j}w_{i, j}(1 - v_i v_j)\\
	\textbf{subject to }\ \ \ & v_i \in \{-1, 1\} \ \forall i \in [n]
\end{align*}
This is a quadratic integer program, so there are no known efficient algorithms to compute an optimal assignment. Instead, the Goemans-Williamson algorithm solves the following relaxation, where $v_i$ and $v_j$ are vectors that are allowed to take on values in the unit sphere $S^{n - 1}$:
\begin{align*}
	\textbf{Maximize }\ \ \ & \frac{1}{2}\sum_{i < j}w_{i, j}(1 - \langle v_i, v_j \rangle)\\
	\textbf{subject to }\ \ \ & v_i \in S^{n - 1} \ \forall i \in [n]
\end{align*}
By defining variables $X_{i, j} := \langle v_i, v_j \rangle$, this maximization problem becomes an SDP, since the constraint that $v_i \in S^{n - 1}$ can be written as the pair of linear constraints $X_{i, i} \leq 1$ and $X_{i, i} \geq 1$, while $X \succeq 0$ if and only if $X = B^\top B$ for some $B$, which happens if and only if each $X_{i, j}$ is the inner product of column $i$ of $B$ with column $j$ of $B$---thus the $v_i$ vectors are precisely the columns of $B$.

The first step is to solve this SDP, which can be done in polynomial time via various different algorithms \cite{SDPPolyTime}. Given a solution $X = B^\top B$, the next step is to extract the matrix $B$. This can be accomplished efficiently via an algorithm known as \emph{incomplete Choleski decomposition} \cite[Alg.\txt{} 4.2.2]{CholeskiRef}. The vectors $\seq{v}{n}$, which are the columns of $B$, define an embedding of the input graph into $n$-dimensional Euclidean space. The final cut is then obtained by splitting these points by a random hyperplane through the origin: choose a random $h \in S^{n - 1}$ and define one side of the cut to be all vertices $i$ such that $\langle h, v_i \rangle \geq 0$. Leveraging the geometry behind this algorithm, one can compute that the expected value of the cut is at least
$$\frac{\alpha_\txt{GW}}{2}\sum_{i < j}w_{ij}(1 - v_i \cdot v_j)$$
where
$$\alpha_\txt{GW} := \min_{0 \leq \theta \leq \pi} \frac{2\theta}{\pi(1 - \cos \theta)} \approx 0.87856.$$
Since the value of the optimal cut is at most the optimal SDP value,
$$\frac{1}{2}\sum_{i < j}w_{ij}(1 - v_i \cdot v_j),$$
this gives an $(\alpha_\txt{GW} - \delta)$-approximation algorithm for any $\delta > 0$, where the $-\delta$ comes from the fact that we cannot solve SDPs exactly, but can solve them up to any arbitrarily small error. (Technically, it is a \emph{randomized} approximation algorithm, though it was subsequently derandomized \cite{GWDerandomization}.)

\subsection{Inapproximability}\label{subInapproximability}
The Goemans-Williamson algorithm was the first improvement from the trivial $\frac12$-approximation algorithm of Sahni and Gonzalez in 19 years, and remains the best known approximation algorithm for \maxcut to date. One might naturally ask, is there any better polynomial time algorithm, achieving an even greater approximation ratio than $\alpha_\txt{GW}$? Unconditionally answering such a question in the negative is hopeless since we cannot even rule out the existence of an efficient algorithm solving \maxcut exactly. Until the $\PP$ vs.\txt{} $\NP$ question is resolved, the best we can hope for is a guarantee that an algorithm is the optimal polynomial time approximation algorithm assuming $\PP \neq \NP$.

The standard technique for showing such so-called \emph{inapproximability} results for a maximization problem $\Lambda$ is to show that the following \emph{gap problem}, written $\gap{c}{s}\Lambda$ for $0 < s \leq c$ (borrowing notation from \cite{UGCSurvey}), is hard: given an instance $I$ of $\Lambda$ in which either
\begin{enumerate}[label={(\arabic*)}]
	\item\label{itmGapProblemCompleteness} the optimal value of $I$ is at least $c$, or
	\item\label{itmGapProblemSoundness} the optimal value of $I$ is less than $s$,
\end{enumerate}
decide which of the two cases \ref{itmGapProblemCompleteness} or \ref{itmGapProblemSoundness} holds. Suppose there existed an $\frac{s}{c}$-approximation algorithm $A$ for $\Lambda$. Then, given an instance $I$ of $\Lambda$, we can run $A$ to compute the approximate value $x$. Since $x$ is the value of some solution, if $x \geq s$ we know we are not in case \ref{itmGapProblemSoundness}, so we must be in case \ref{itmGapProblemCompleteness}. Otherwise, if $x < s$, we know that the optimal value is at most $\frac{c}{s} x < \frac{c}{s} s = c$, so we cannot be in case \ref{itmGapProblemCompleteness}, and hence must be in case \ref{itmGapProblemSoundness}. Thus, we can use $A$ to decide the gap problem in polynomial time, so if deciding the gap problem is $\NP$-hard, then there does \emph{not} exist an $\alpha$-approximation algorithm for any $\alpha \geq \frac{s}{c}$ unless $\PP = \NP$. The ratio $\frac{s}{c}$ is called the \emph{gap ratio}.

The most famous result of this kind is the PCP Theorem \cite[Chapter 11]{AroraBarak}, which gives a reduction from \threesat to \gap{1}{s}\threesat, mapping satisfiable formulas to satisfiable formulas, and unsatisfiable formulas to formulas in which no more than an $s$ fraction of clauses can be simultaneously satisfied, for a universal constant $s < 1$. Further such \emph{gap-preserving} reductions have been discovered from that gap problem to other gap problems, implying many useful inapproximability bounds \cite[Sec.\txt{} 2.3]{UGCSurvey}. Some of these bounds are tight in the sense that there are known algorithms attaining those bounds, while for other problems, there is still a gap in our knowledge.

\subsection{Unique games}\label{subUGBackground}
A central research question in the study of approximation algorithms concerns the approximability of a certain constraint satisfaction problem, called \ug. For any positive integer $q$, \UG($q$) is the unweighted CSP where the domain has size $q$ and constraints may be imposed between pairs of variables such that the value of one variable uniquely determines the value of the other. It is often convenient to think of \ug instances as being defined on some graph $G$, where the vertices represent variables and each edge $\{u, v\} \in E(G)$ has a permutation $\pi_{u, v}$ on the label set $[q]$ defining which labels for vertex $u$ correspond to which labels for vertex $v$. The goal is to label the vertices with elements from $[q]$ to be consistent with a maximum number of permutations. For example, Figure \ref{figUGExample} shows a \UG(2) instance with edge permutations written in cycle notation, along with one of the optimal vertex labelings, satisfying $\frac{3}{4}$ of the constraints.

\ipns{1.8}{UGExampleCropped}{A \ug instance over the label set $\{1, 2\}$.}{\label{figUGExample} A \ug instance over the label set $\{1, 2\}$ represented graphically, along with one optimal solution (green). Only the bottom edge (red) is unsatisfied by this solution.}

The uniqueness property of the constraints makes it easy to determine whether an instance is completely satisfiable. Just pick any vertex and enumerate all of the $k$ possible labels for it. For each label, inductively derive the unique labels of neighbouring vertices, until labels for the entire graph have been determined. The instance is completely satisfiable if and only if some initial choice of label for the first vertex makes all edges consistent with the derived labels. If the graph is disconnected, then repeat this algorithm for each connected component.

However, when the input instance is not completely satisfiable, there are no known good algorithms for approximating the maximal fraction of simultaneously satisfiable constraints. Indeed, it has been shown that for any $\delta > 0$, there exists a $q$ such that \gap{\frac12}{\delta}\UG($q$) is $\NP$-hard \cite[Theorem I.5]{UGMidGap2018-2}. As a consequence, for a sufficiently large label set, it is impossible to approximate \ug to within any constant factor. The \emph{Unique Games Conjecture (UGC)} is a strengthening of this statement:\thmspace

\begin{conjecture}[Unique Games Conjecture]\label{cnjUGC}
	For all $\varepsilon, \delta > 0$, there exists a positive integer $q$ such that deciding \gap{1-\varepsilon}{\delta}\textup{\UG($q$)} is $\NP$-hard. In other words, for a large enough label set, it is $\NP$-hard to distinguish instances in which at least $1 - \varepsilon$ constraints can be satisfied from instances in which less than $\delta$ constraints can be satisfied.
\end{conjecture}

Unlike the other central open problems in complexity theory, UGC is a conjecture built upon another conjecture. It is not claiming that it is ``impossible" to solve the \ug gap problem; rather, that it is $\NP$-hard. Thus, even if the UGC is proven, one will still not be able to make any unconditional claims about the nonexistence of algorithms until $\PP \neq \NP$ is proved as well. This is one reason it is believed that the UGC will be resolved sooner than other longstanding open questions.

If it is true (as most researchers suspect), many other inapproximability results would immediately follow. Khot, Kindler, Mossel and O'Donnell \cite{KKMO} show that, assuming the UGC, it is $\NP$-hard to approximate \maxcut to within any factor greater than $\alpha_\txt{GW}$. In other words, this seemingly arbitrary constant obtained by the geometric analysis of the Goemans-Williamson algorithm is, in fact, \emph{the} optimal approximation ratio, a fundamental constant for the \maxcut problem. Subsequently, Raghavendra \cite{RagThesis} discovered a general explanation for this surprising fact, proving that, for \emph{every} CSP, there is a polynomial time SDP-based algorithm that gives the optimal approximation ratio assuming the UGC (and $\PP \neq \NP$).

In Chapter \ref{chaUG} we consider a special subclass of \ug instances which we call \gug instances. These are instances with the following additional properties:
\begin{enumerate}[label={(\arabic*)}]
	\item\label{itmGUGVertexLabels} The set of labels is identified with a finite Abelian group $A$.
	\item\label{itmGUGEdgeLabels} For every edge permutation $\pi$, there is some $g \in A$ such that $\pi(x) = g + x$ (we always write the group operation additively). Thus, we can identify the set of edge permutations with $A$ as well.
\end{enumerate}
By restricting the UG optimization problem in this way, one might hope that it is easier to solve. However, it turns out that the problem \maxtwolin($q$), in which $A = \zz / q\zz$, is just as hard as the general UG problem in the sense that there is a gap-preserving reduction from \UG($k$) to \maxtwolin($q$) \cite{KKMO} for arbitrary $k$ and $q$. This reduction implies that, by replacing \UG($k$) by \maxtwolin($k$) we get a conjecture which is logically equivalent to the UGC.


\section{Background on logic}\label{secFPCBackground}

This thesis concerns approximating CSPs not with polynomial time algorithms but with sentences of logic. What does this mean? It turns out that there is a natural way in which a description of a computational problem in a formal logic can be translated into an algorithm for solving it. Conversely, algorithms for solving problems can, in many cases, be translated back into logical sentences describing the problem being solved. Under this correspondence, the resources used by the algorithm (time, space, nondeterminism, etc.) correspond to the operators which are allowed by the logic (different kinds of quantifiers, inductive definitions, etc.). The study of this correspondence is known as \emph{descriptive complexity}, a more detailed and complete discussion of which is the topic of several books; see \cite{DCBook}, for example. Here we only give a brief introduction to the ideas and notation of descriptive complexity that are needed for this thesis.

\subsection{Descriptive complexity and FPC}\label{subWhyFPCMatters}
Sentences of logic and the structures whose properties they describe are all defined with respect to a \emph{vocabulary}, or \emph{signature}, which enumerates special symbols that are necessary to talk about a given problem. All vocabularies we consider are \emph{relational}, meaning the only special symbols allowed are for relations (no function or constant symbols). A relational vocabulary $\tau$ takes the form
$$\tau := \langle R_1^{a_1}, R_2^{a_2}, \dots, R_m^{a_m}\rangle,$$
where $a_i$ is the \emph{arity} of \emph{relation symbol} $R_i$ (the arity is sometimes dropped when obvious from context). A \emph{$\tau$-structure} consists of a nonempty set, called its \emph{universe}, together with relations instantiating, or \emph{interpreting}, as it is commonly called, the symbols of the signature $\tau$. To make a programming analogy, if a vocabulary $\tau$ is a type, then a $\tau$-structure is a term of that type. A $\tau$-structure $\mathbb{A}$ with universe $A$ is written as
$$\mathbb{A} := \langle A, R_1^{\mathbb{A}}, R_2^{\mathbb{A}}, \dots, R_m^{\mathbb{A}} \rangle,$$
where each $R_i^\mathbb{A}$ is a relation over $A$ of arity $a_i$.

For example, the vocabulary of graphs, $\tau_\textup{graph}$, consists of a single binary relation, $E^2$. A graph can be encoded as a $\tau_\textup{graph}$-structure $\mathbb{G} = \langle V, E^\mathbb{G} \rangle$ in which the universe $V$ is the set of vertices and the edge relation symbol $E$ is interpreted by a symmetric relation $E^\mathbb{G}$ consisting of all the edges.

The fundamental question asked in descriptive complexity is, given some vocabulary $\tau$ and some decision problem $\mathcal{P}$ concerning $\tau$-structures, what logic is needed to express $\mathcal{P}$? For example, the problem of deciding whether a graph is 2-regular (meaning all vertices have exactly 2 neighbours) is expressible in \emph{first order (FO) logic}, by the sentence
$$\forall u \exists v_1 \exists v_2 (\neg (v_1 = v_2)) \wedge E(u, v_1) \wedge E(u, v_2) \wedge \forall v_3 \ E(u, v_3) \to (v_3 = v_1 \vee v_3 = v_2).$$

Any property expressible in FO logic can be decided in polynomial time (as a function of the size of the universe of the input structure), but it is not the case that any polynomial time decidable property is expressible in FO logic---for example, there is no sentence of FO logic expressing the property that a graph is connected. It is therefore said that FO logic does not \emph{capture} the complexity class $\PP$. Thus, several extensions to FO logic have been proposed. \emph{LFP} is an extension of FO logic allowing for relational variables to be defined inductively, in such a way that LFP-definable properties can still be computed in polynomial time (see \cite[Chapter 4]{DCBook}). If we assume that the input structure is \emph{ordered}, meaning that the vocabulary comes with a binary relation ``$\leq$" interpreted as a total order on the universe, then LFP captures $\PP$. This important result is known as the Immerman-Vardi Theorem \cite{ImmermanVardi1, ImmermanVardi2}.

Over \emph{unordered} structures, however, LFP is not even able to express the simple property that the input structure has an even number of elements in its universe. For this, we can augment LFP with the ability to quantify over numeric variables (taking values from 0 up to the size of the input), along with \emph{counting quantifiers}, which assert that a certain number of objects satisfy a certain predicate (see \cite[Sec.\txt{} 12.3]{DCBook}). The resulting logic is called \emph{fixed point logic with counting (FPC)}. While FPC-definable properties are decidable in polynomial time, FPC still does not capture $\PP$, but counterexamples are highly nontrivial. Nevertheless, a wide range of powerful algorithmic techniques, including linear and semidefinite programming, \emph{are} expressible in FPC, making it an important logic to study.

Since the order of the elements assumed by the Immerman-Vardi theorem can be completely arbitrary, the only power gained from assuming an order relation is the ability to repeatedly choose arbitrary elements. In the absence of an ordering relation, indistinguishable elements must be treated equally. Thus, FPC encapsulates the idea of symmetric computation. Polynomial time algorithms correspond to FPC sentences only if they do not make arbitrary, symmetry-breaking choices. A non-example is solving systems of linear equations over finite fields. The standard algorithm for this problem is Gaussian elimination, which requires one repeatedly choose a pivot. In the presence of an ordering of the rows and columns of a matrix, one can choose the nonzero entry in the least column of the ordering, breaking ties by choosing the least row. Without an order, any such choice would break the symmetry of the input problem, so the only thing a symmetric algorithm could do is to try \emph{all} possible pivots at each iteration, which would take an exponential amount of time. Indeed, it has been shown that solving systems of linear equations (over finite fields) is not definable in FPC  \cite{SolvingEquationsNotInFPC}. It is difficult to rigorously define exactly what is meant by ``symmetry breaking," though hopefully the intuition is clear. Anderson and Dawar \cite{SymmetricCircuits} give a precise instantiation of this meta-observation defined in terms of symmetric circuits.

\subsection{Finite structures for CSPs}\label{subFiniteStructures}
To represent an unweighted CSP as a relational structure, we use a vocabulary consisting of relation symbols $\seq{P}{m}$, one for each kind of constraint of the problem. An instance $\mathbb{A} = \langle A, P_1^\mathbb{A}, P_2^\mathbb{A}, \dots, P_m^\mathbb{A} \rangle$ has a universe $A$ consisting of the set of variables, where each $P_i^\mathbb{A}$ defines the set of tuples of variables to which the constraint $P^i$ is applied. For example, in \threesat, there are $m = 8$ different kinds of constraints (clauses), all of arity 3, where $P_1$ is for clauses of the form $(x_1 \vee x_2 \vee x_3)$, $P_2$ is for clauses of the form $(x_1 \vee x_2 \vee \overline{x_3})$, and so on. For \ug on a label set of size $q$, there is one constraint $P_\pi$ of arity 2 for each permutation $\pi: [q] \to [q]$. We call this vocabulary\footnote{Note that $\tauuug$ is for \emph{unweighted} \ug instances only, which breaks a notational convention used throughout this thesis: for any CSP $\Lambda$ other than \UG($q$), $\tau_\Lambda$ is the vocabulary of \emph{weighted} instances of $\Lambda$. } $\tauuug$.

To represent a weighted CSP, we first have to introduce some extra machinery to deal with numbers. This definition is loosely based on the structures used by Dawar and Wang \cite{SDPInFPC1} to represent vectors and matrices. We can represent a natural number $n$ as a relational structure $\mathbf{n} = ([b], B^\mathbf{n})$ in the vocabulary $\tau_\nn := \langle B, \leq \rangle$. The universe $[b]$ has size $b = \lceil\log_2(n + 1)\rceil$, $\leq$ is a binary relation interpreted as the usual linear order on $[b]$ (from least significant bits to most significant bits), and $B$ is a unary relation encoding the bit representation of $n$, i.e.,
$$B^\mathbf{n} := \{k \in [b] \suchthat \txt{the $k\tth$ (least significant) bit of $n$ in binary is 1}\}.$$
To represent negative integers, we add a new unary relation symbol $S$ to $\tau_\nn$ to obtain a new vocabulary $\tau_\zz$, where $n$ is positive if and only if $S^\mathbf{n} = \emptyset$. To represent rational numbers, we replace $B$ with unary relation symbols $N$ and $D$ for the numerator and denominator (which act in the same way as $B$) to obtain a new vocabulary $\tau_\qq$.

Let $\Lambda$ be a CSP. To represent a weighted instance of $\Lambda$, we have to combine variables and numbers together. That is, we use what is known as a \emph{two-sorted universe}, in which there are two different kinds of elements, in this case a variable sort $T$ and a number sort $[b]$, where
$$b := \left\lceil \log_2\left(1 + (\txt{max numerator or denominator of any constraint weight})\right) \right\rceil.$$
The vocabulary $\tau_\Lambda$ consists of the usual order relation $\leq$ which is interpreted by the instance $\mathbb{A}$ as a total order on $[b]$ and relation symbols $N_i$, $D_i$ and $S_i$ of arity $r_i + 1$ for each constraint type $i$ of arity $r_i$, where $N_i$ is interpreted as 
\begin{align*}
	N_i^\mathbb{A} := \{(\mathbf{x}, k)&\in T^{r_i} \times [b] \suchthat \txt{the $k\tth$ bit of the numerator of}\\&\txt{ the weight of constraint $i$ applied to tuple $\mathbf{x}$ in binary is 1}\},
\end{align*}
$D_i$ is like $N_i$, but for the denominator, and $S_i^\mathbb{A}(\mathbf{x}, \cdot)$ is empty if and only if the weight of constraint $i$ applied to $\mathbf{x}$ is negative. Having negative weights allows us to consider minimization problems as well as maximization problems, matching the framework for CSPs developed by Raghavendra \cite{RagThesis}; the objective is always to maximize the total weight. Note that the order relation is only imposed on the bit positions, not the variables, so we can still represent unordered structures without breaking symmetry.

\subsection{Table of signatures}
For reference, Table \ref{tabSignatures} lists several important signatures used throughout this thesis. The last two signatures, $\tau_\txt{mat}$ and $\tau_\txt{SDP}$, are introduced in Section \ref{secGWAlgoInFPC}. In all cases, $\leq^2$ is interpreted as an order on bit positions only, not on variables or abstract indices.

\begin{table}[H]\begin{center}\begin{tabular}{|c|c|l|}
	\hline
	Signature & Relation symbols & Used to represent\\
	\hline\hline
	$\tau_\nn$ & $\leq^2$, $B^1$ & Natural numbers\\
	\hline
	$\tau_\zz$ & $\leq^2$, $B^1$, $S^1$ & Integers\\
	\hline
	$\tau_\qq$ & $\leq^2$, $N^1$, $D^1$, $S^1$ & Rationals\\
	\hline
	$\tau_\textup{graph}$ & $E^2$ & Graphs\\
	\hline
	$\tauuug$ & $P_\pi^2$ for $\pi: [q] \to [q]$ & Unweighted UG instances\\
	\hline
	$\tau_\Lambda$ & \thead{$\leq^2$; $N_i^{r_i + 1}$, $D_i^{r_i + 1}$, $S_i^{r_i + 1}$\\\textup{for each constraint}\\\textup{type $i$ of arity $r_i$}} & Weighted $\Lambda$ instances\\
	\hline
	$\tau_{\txt{\maxcut}}$ & $\leq^2$, $N_1^3$, $D_1^3$, $S_1^3$ & Weighted \maxcut instances\\
	\hline
	$\tau_{\txt{mat}}$ & $\leq^2$, $X^3$, $D^3$, $S^3$ & Matrices\\
	\hline
	$\tau_{\txt{SDP}}$ & \thead{$\leq^2$, $X_A^4$, $D_A^4$, $S_A^4$, $X_b^2$,\\ $D_b^2$, $S_b^2$, $X_C^3$, $D_C^3$, $S_C^3$} & Semidefinite programs\\
	\hline
\end{tabular}\end{center}\caption[Table of signatures.]{\label{tabSignatures}%
Table of signatures. Note that $\tau_{\txt{\maxcut}}$ and $\tau_{\txt{mat}}$ are the same up to renaming. We write $X$ instead of $N$ in $\tau_{\txt{mat}}$ just to be consistent with the notation of Dawar and Wang \cite{SDPInFPC1}.}\end{table}

\subsection{Lower bounds for FPC}\label{subPebblingGameDefinition}
To show that a property $\mathcal{P}$ is definable in FPC, we just need to exhibit a single FPC sentence $\phi$ and prove that a structure $\mathbb{A}$ satisfies $\phi$ (written $\mathbb{A} \models \phi$) if and only if $\mathbb{A}$ has property $\mathcal{P}$. Showing that a property is \emph{not} definable in FPC is trickier, since we must argue that no such sentence works. The standard proof technique is to assume, for the sake of contradiction, that there was such a sentence $\phi$ defining $\mathcal{P}$. Then there exists a $k$ such that $\phi$ can be translated into $C^k$, the fragment of \emph{infinitary} FO logic with counting quantifiers consisting of (possibly infinite) sentences with only $k$ variables \cite{InfinitaryTranslation}. We denote the minimum such $k$ by $\mu(\phi)$. To show the contradiction, we construct a pair of structures $\mathbb{A} = \mathbb{A}_k$ and $\mathbb{B} = \mathbb{B}_k$ such that $\mathbb{A}$ has property $\mathcal{P}$ but $\mathbb{B}$ does not, yet \emph{any sentence of $C^k$ cannot distinguish $\mathbb{A}$ from $\mathbb{B}$}, in the sense that $\mathbb{A}$ satisfies any $C^k$ sentence if and only if $\mathbb{B}$ does. When this is the case, we write $\mathbb{A} \equiv_{C^k} \mathbb{B}$.

There is a useful characterization of the relation $\equiv_{C^k}$ in terms of a game between two players, Spoiler and Duplicator, called the \emph{$k$-pebble bijective game}. The board on which they play consists of the universe $A$ of structure $\mathbb{A}$ and the the universe $B$ of structure $\mathbb{B}$. Spoiler's objective is to prove that the structures are different, while Duplicator's objective is to pretend that they are the same. There are $k$ pairs of \emph{pebbles}, initially not placed anywhere. Throughout the game, the pairs of pebbles will be placed on elements of the two universes, one pebble in each universe. Each round of the game consists of three parts:
\begin{enumerate}[label={(\arabic*)}]
	\item\label{itmPebbleGame1} Spoiler picks up one of the $k$ pairs of pebbles, removing them from the board.
	\item\label{itmPebbleGame2} Duplicator gives a bijection $f: A \to B$ such that, for all $1 \leq i \leq k$, if the $i\tth$ pebble pair is placed on some pair of elements $a_i \in A$, $b_i \in B$, then $f(a_i) = b_i$.
	\item\label{itmPebbleGame3} Spoiler places the pebbles back down, placing one pebble on some $a \in A$ and the other pebble on $f(a) \in B$.
\end{enumerate}
At the end of a round, Spoiler wins if the map sending each pebbled element in $A$ to its correspondingly-pebbled element in $B$ is not a \emph{partial isomorphism} between the two structures, i.e., there is some relation in one of the two structures that holds of a set of pebbled elements, but the corresponding relation does not hold in the other structure of the correspondingly-pebbled elements. If Spoiler is unable to win the game in any finite number of moves, then Duplicator wins.\thmspace

\begin{theorem}[Hella \cite{BijectiveGame}]
	Duplicator wins the $k$-pebble bijective game played on $\mathbb{A}$ and $\mathbb{B}$ if and only if $\mathbb{A} \equiv_{C^k} \mathbb{B}$.
\end{theorem}

So, to show that two structures are indistinguishable, and thus that $\phi$ does not express $\mathcal{P}$, we just need to present a winning strategy for Duplicator.

\subsection{Interpretations}\label{subInterpretations}
So far we have only discussed sentences of logic acting as algorithms for decision problems. If a logical sentence corresponds to a Turing machine, then whether a structure satisfies the sentence corresponds to whether the Turing machine accepts the encoding of that structure. For some applications, however, it is useful to consider Turing machines which output something more complicated than ``accept" or ``reject." The logical analogue of such a machine is called an \emph{interpretation}.

To construct an interpretation $\Theta$, suppose we are given an input structure $\mathbb{A}$ in some signature $\sigma$, and wish to define the output $\mathbb{B} = \Theta(\mathbb{A})$ in some potentially different signature $\tau$. First we must define the universe of $\mathbb{B}$ in terms of the universe of $\mathbb{A}$. This can be done by taking the universe of $\mathbb{B}$ to be the set of $d$-tuples of elements of $\mathbb{A}$ satisfying some FPC formula of $d$ free variables, written in the vocabulary $\sigma$. By choosing $d$ large enough, we can define universes of size up to $n^d$, where $n$ is the size of the input structure. Next, we must define each relation symbol appearing in $\tau$. For a symbol $R_i$ of arity $a_i$, we must define on which $a_i$-tuples of elements of $\mathbb{B}$, i.e., $a_i$-tuples of $d$-tuples of elements of $\mathbb{A}$, the relation $R_i^\mathbb{B}$ holds. This can be accomplished via a FPC formula of $da_i$ free variables, again written in the vocabulary $\sigma$, where we take the relation to hold if and only if the formula is satisfied.

Thus, a \emph{$d$-ary FPC-interpretation of $\tau$ in $\sigma$} is defined by a finite sequence of FPC formulas in the vocabulary $\sigma$, as outlined above \cite[Sec.\txt{} 2.2]{DefinableInapproximabilityJournal}. For an interpretation $\Theta$, we define $\mu(\Theta)$ to be the maximum value of $\mu(\phi)$ for any formula $\phi$ of $\Theta$.

Just as polynomial time reductions can be used to transfer computational hardness results from one problem to another, interpretations can transfer logical inexpressibility results from one problem to another: if some property $\mathcal{P}$ of $\sigma$-structures is \emph{not} definable in FPC, and there is an interpretation $\Theta$ of $\tau$ in $\sigma$ such that a $\sigma$-structure $\mathbb{A}$ has property $\mathcal{P}$ if and only if the $\tau$-structure $\Theta(A)$ has property $\mathcal{P'}$, then $\mathcal{P'}$ is not definable in FPC either.

Interpretations also give us a useful way to define what it means to ``solve" an optimization problem. Recall that, for a CSP $\Lambda$, $\tau_\Lambda$ is the vocabulary of weighted instances of $\Lambda$, and $\tau_\qq$ is the vocabulary of rational numbers. By an \emph{FPC-definable algorithm} for $\Lambda$ we mean an interpretation $\Theta$ of $\tau_\qq$ in $\tau_\Lambda$ such that, for any $\tau_\Lambda$-structure $\mathbb{A}$, the optimal value of $\mathbb{A}$ is equal to $\Theta(\mathbb{A})$.


\section{Definable inapproximability}\label{secDefinableInapproximabilityReview}

We are interested not just in solving CSPs exactly in FPC, but in approximating them. Bringing together our earlier definitions of FPC-definable algorithm and approximation algorithm, we say that, for a CSP $\Lambda$ and for $0 \leq \alpha \leq 1$, an \emph{FPC-definable $\alpha$-approximation algorithm} for $\Lambda$ is an FPC-interpretation $\Theta$ of $\tau_\qq$ in $\tau_\Lambda$ such that, for any $\tau_\Lambda$-structure $\mathbb{A}$ of optimal value $x^*$, $\alpha x^* \leq \Theta(\mathbb{A}) \leq x^*$.

In a recent paper, Atserias and Dawar \cite{DefinableInapproximabilityJournal} prove the first (to the author's knowledge) inapproximability result for FPC. Their main construction is a pair of \threexor instances (like \threesat except with XORs in place of ORs between literals) $\mathbb{A}_k$ and $\mathbb{B}_k$, for any $k$, such that $\mathbb{A}_k$ is completely satisfiable, $\mathbb{B}_k$ is only $\frac12 + \delta$ satisfiable (for arbitrarily small $\delta$), but $\mathbb{A}_k \equiv_{C^k} \mathbb{B}_k$. As a consequence, there is no FPC-definable $\alpha$-approximation algorithm for $\alpha > \frac12$, for if there was such an FPC-interpretation $\Theta$, if we let $\delta$ be such that $\frac12 + \delta < \alpha$ and let $k := \mu(\Theta)$, we would necessarily have $\Theta(\mathbb{A}_k) = \Theta(\mathbb{B}_k)$, which contradicts the requirements that $\Theta(\mathbb{A}_k) \geq \alpha n$ and $\Theta(\mathbb{B}_k) \leq (\frac{1}{2} + \delta) n$ (where $n$ is the total number of constraints). This is analogous to showing that a gap problem is hard---in this case, $s = \frac12 + \delta$ and $c = 1$, so the gap ratio is $(\frac12 + \delta) / 1 = \frac12 + \delta$. Atserias and Dawar then show that several existing gap-preserving reductions from \threexor to other problems could be cast as FPC-interpretations, resulting in FPC inapproximability bounds for \threesat, \vc and \labelcover.


\chapter{Approximating constraint satisfaction problems in FPC}\label{chaCSPs}

In this section we prove an FPC-analogue of Raghavendra's result \cite{RagThesis} that, assuming the UGC, the optimal approximation algorithm for any CSP is obtained by rounding a specific SDP relaxation. The proof consists of verifying, firstly, that Raghavendra's general polynomial time algorithm is definable as an FPC-interpretation of $\tau_\qq$ in $\tau_\Lambda$, and secondly, that Raghavendra's reduction from \ug to $\Lambda$ is definable as an FPC-interpretation of $\tau_{\Lambda}$ in $\tauuug$. We begin by discussing the special case of \maxcut.


\section{An FPC sentence approximating \maxcut}\label{secGWAlgoInFPC}

To translate the Goemans-Williamson algorithm into an FPC-interpretation, we must first understand how to translate its core subroutine: solving a semidefinite program. This is studied by Dawar and Wang \cite{SDPInFPC1}, who define a vocabulary $\tau_\txt{SDP}$ for SDP instances and an FPC-interpretation which approximately solves them. Like the vocabulary $\tau_\Lambda$ for weighted CSPs, SDPs are defined over a multi-sorted universe, with an unordered sort for indexing the rows and columns of matrices, another unordered sort for indexing the constraints and, finally, an ordered sort for representing numbers in binary. There are 10 relation symbols,
$$\tau_{\txt{SDP}} := \langle \leq^2, X_A^4, D_A^4, S_A^4, X_b^2, D_b^2, S_b^2, X_C^3, D_C^3, S_C^3 \rangle,$$
encoding the constraint matrices $\{A_k\}$, corresponding constraint vectors $\{b_k\}$ and objective matrix $C$. The $X$ relations encode the numerators, the $D$ relations encode denominators, and the $S$ relations encode signs. For example, if a tuple $(k, i, j, m)$ is in the relation $D_A^\mathbb{A}$ it means that, in the $k\tth$ constraint matrix of SDP $\mathbb{A}$, the $m\tth$ bit of the numerator of the entry at row $i$, column $j$ is a 1. If the unary relation $S_C^\mathbb{A}(i, j, \cdot)$ is nonempty, it means that the entry of the objective matrix of $\mathbb{A}$ at row $i$, column $j$ is negative. As usual, $\leq$ encodes the total order on the bit sort.

The output of an SDP solver is the matrix of optimal variable values, which is encoded in the vocabulary
$$\tau_{\txt{mat}} := \langle \leq^2, X^3, D^3, S^3 \rangle,$$
similarly as in the encoding of $C$.\thmspace
\begin{theorem}[Dawar and Wang \cite{SDPInFPC1}]\label{thmSDPInFPC}
	There is an FPC-interpretation $\Phi$ of $\tau_\txt{mat}$ in $\tau_\txt{SDP} \ \dot{\cup} \ \tau_\qq$ such that, given a bounded and feasible SDP $\mathbb{A}$ (encoded as a $\tau_{\txt{SDP}}$-structure) and some $\delta > 0$ (encoded as a $\tau_\qq$-structure), $\Phi(\mathbb{A}, \delta)$ encodes a matrix $X$ which is within $\delta$ of a feasible solution to $\mathbb{A}$ (e.g., in the $L^2$-norm), and has value within $\delta$ of an optimal solution.
\end{theorem}

Notice how it is crucial that the index sets are unordered. If the entire universe was ordered, then the statement would follow immediately from the Immerman-Vardi Theorem, but would be useless. For example, in the SDP for \maxcut, the rows and columns of the matrices, as well as the linear constraints, correspond to vertices. If there was a relation symbol in $\tau_\txt{SDP}$ which must encode the order on these rows, columns and constraints, then, to define the SDP from the \maxcut instance we would have to define that order, which is impossible if the vertices are not ordered to begin with. Since there is no such order on the index and constraint sorts, we may simply define the row sort to be the vertex set, and so on.

Setting up the rest of the \maxcut SDP is easy, but tedious, so here we just go through one example. Consider the task of defining the relations $X_C$ and $D_C$. Recalling the SDP for \maxcut defined in Section \ref{subSDPAndGW}, for $i < j$, the coefficient in row $i$, column $j$ of the objective matrix is $-\frac{1}{2}w_{ij}$. Thus, an arbitrary entry $(i, j, m)$ is in the relation $X_C$ (respectively, $D_C$) if and only if the $m\tth$ bit of the numerator (respectively, denominator) of $-\frac{1}{2}w_{ij}$ is a 1. Recall that, in the encoding for \maxcut instances, the weights are encoded as ternary relations $N_1$ and $D_1$ expressing the numerator and denominator in binary. Since multiplying by $\frac{1}{2}$ is the same as shifting the bits of the denominator up by 1, we may define $X_C(i, j, m)$ to be true if and only if
\begin{equation}\label{equInterpretationFormulaExample1}
	N_1(i, j, m)
\end{equation}
is true, and define $D_C(i, j, m)$ be true if and only if $D_1(i, j, m - 1)$ is true, i.e.,
\begin{equation}\label{equInterpretationFormulaExample2}
	\exists m_1\ m_1 \leq m \wedge (\neg \exists m_2\ m_1 \leq m_2 \wedge m_2 \leq m) \wedge D_1(i, j, m_1)
\end{equation}
is true. In this case, both (\ref{equInterpretationFormulaExample1}) and (\ref{equInterpretationFormulaExample2}) are FO formulas; all that we require is that they be FPC formulas. More complicated arithmetical operations can be translated into FPC formulas as well \cite[Sec.\txt{} 3.3]{HolmThesis}, so we ignore these details hereafter.

Thus, we have an interpretation of $\tau_\txt{SDP}$ in $\tau_{\txt{\maxcut}}$, which we can compose with the interpretation $\Phi$ from Theorem \ref{thmSDPInFPC} to obtain an approximately optimal solution matrix $X$. The next steps of the Goemans-Williamson algorithm are to find a matrix $B$ such that $X = B^\top B$ and pick a random hyperplane $h$. From $h$ and $B$, we would then be able to define the two sets of the cut, and from that, the value of the cut. The first difficulty is that the incomplete Choleski decomposition algorithm for extracting $B$ contains symmetry-breaking steps, as it is similar to Gaussian elimination. The difficulty runs even deeper though. In fact, the whole approach to these latter steps of the algorithm is unattainable in FPC, since merely defining a cut at some intermediate step would break symmetry. For instance, if the input is a set of size $2n$ containing edges of nonzero weight only between $n$ disjoint pairs of vertices, then there are $2^n$ optimal cuts in the graph with automorphisms taking any one to any other. Since FPC-interpretations must respect automorphisms of the input structure, if some sentence of FPC was able to select one of these cuts, it would have to simultaneously select \emph{all} of them. This is impossible, since the output of an FPC-interpretation necessarily has polynomial size.

Therefore, we must compute the optimal value \emph{without ever computing a specific cut}, or even computing a specific valid matrix $B$, for that matter. While we may not be able to compute in FPC the exact value returned by the algorithm, we can at least use Goemans' and Williamson's analysis to bound it. As mentioned in Section \ref{subSDPAndGW}, the expected value of the cut returned by the algorithm is at least
$$\frac{\alpha_\txt{GW}}{2}\sum_{i < j}w_{ij}(1 - v_i \cdot v_j) = \frac{\alpha_\txt{GW}}{2}\sum_{i < j}w_{ij}(1 - X_{i, j}).$$
Fortunately, this quantity \emph{is} definable in FPC, since we have already constructed FPC definitions for $w_{i, j}$ (part of the input) and $X_{i, j}$ (coming from the interpretation $\Phi$ of Theorem \ref{thmSDPInFPC}). The rest is just simple arithmetic, so each bit can be defined by an FPC formula.

Thus, we have shown the following result, which completely parallels Goemans' and Williamson's result for polynomial time computation.\thmspace
\begin{theorem}\label{thmMaxCutInFPC}
	For any $\varepsilon > 0$, there is an FPC-definable\\ $(\alpha_\txt{GW} - \varepsilon)$-approximation algorithm for \maxcut.
\end{theorem}
The $-\varepsilon$ factor comes from the fact that, in applying Theorem \ref{thmSDPInFPC}, we cannot solve the SDP exactly, but only up to an additive $\delta$. By choosing $\delta$ sufficiently small, we can ensure by continuity that the approximation ratio is at least $\alpha_\txt{GW} - \varepsilon$.


\section{General algorithm for CSPs}\label{secRagAlgo}

The translation into FPC of Raghavendra's \cite{RagThesis} general approximation algorithm for any CSP is similar in essence to that of \maxcut. We first define an SDP from a given instance, then apply Theorem \ref{thmSDPInFPC}, then extract the optimal value. Throughout the remainder of this chapter, we make the simplifying assumption (as Raghavendra does) that CSP instances are normalized so that the sum of all weights is in $[-1, 1]$.

Given an instance $I$ of CSP $\Lambda$, with variable set $\mathcal{V}$, domain $[q]$ and constraint set $\mathcal{P}$, Raghavendra's algorithm defines and solves an SDP called the \emph{\LC relaxation} \cite[Sec.\txt{} 4.5]{RagThesis}. For any constraint $P \in \mathcal{P}$, let $\mathcal{V}(P)$ denote the set of variables appearing in constraint $P$, and $\weight_I(P)$ denote the weight of that constraint. The variables of the SDP consist of a set of $(q \cdot \abs{\mathcal{V}})$-dimensional vectors,
$$\{\mathbf{b}_{i, a} \suchthat i \in \mathcal{V},\ a \in [q]\},$$
and a set of probability distributions over local assignments of variables within each constraint,
$$\{\mu_P \suchthat P \in \mathcal{P}\}.$$
That is, each $\mu_P$ variable is a distribution over $[q]^{\mathcal{V}(P)}$. The \LC relaxation\footnote{Raghavendra's original \LC relaxation looks slightly different because it is written using a more general notation, in which constraints are arbitrary ``payoff" functions from assignments to values in $[-1, 1]$. The SDP written here is what results when constraints are merely ``satisfied" or ``unsatisfied," with satisfied constraints yielding payoffs equal to their weights. } is as follows:

\begin{align*} 
	\textbf{Maximize }\ \ \ & \sum_{P \in \mathcal{P}}\ \weight_I(P) \underset{f \sim \mu_P}{\mathbb{P}} \{f \txt{ satisfies } P\}\\
	\textbf{subject to }\ \ \ & \langle \mathbf{b}_{i, a}, \mathbf{b}_{j, b} \rangle = \underset{f \sim \mu_P}{\mathbb{P}}\{f(i) = a,\ f(j) = b\} \push\hspace{3.44cm}\forall P \in \mathcal{P},\ i, j \in \mathcal{V},\ a, b \in [q]\\
	& \mu_P \in \triangle([q]^{\mathcal{V}(P)}) \hspace{2cm} \forall P \in \mathcal{P}
\end{align*}

At first, the \LC relaxation may look like an ordinary quadratic program, yet it is implicitly an SDP. To show that Raghavendra's algorithm can be defined in FPC, however, we have to put this SDP into the explicit form required by Theorem \ref{thmSDPInFPC}.

First observe that the probability distributions $\mu_P$ can be defined as sets of numbers 
$$\{\mu_P(f) \suchthat P \in \mathcal{P},\ f: \mathcal{V}(P) \to [q]\}$$
summing to 1. Thus, we may rewrite the \LC relaxation as:

\begin{align}\label{equLC2} 
	\textbf{Maximize }\ \ \ & \sum_{P \in \mathcal{P}}\ \sum_{\substack{f: \mathcal{V}(P) \to [q]\\\txt{satisfying } P}} \weight_I(P)\mu_P(f)\\
	\textbf{subject to }\ \ \ & \langle \mathbf{b}_{i, a}, \mathbf{b}_{j, b} \rangle = \sum_{\substack{f: \mathcal{V}(P) \to [q]\\f(i) = a,\ f(j) = b}} \mu_P(f) && \forall P \in \mathcal{P},\ i, j \in \mathcal{V},\ a, b \in [q]\nonumber\\
	& \mu_P(f) \geq 0 && \forall P \in \mathcal{P}.\ f: \mathcal{V}(P) \to [q]\nonumber\\
	& \sum_{f: \mathcal{V}(P) \to [q]} \mu_P(f) = 1 && \forall P \in \mathcal{P}\nonumber
\end{align}

Since there are two kinds of variables, vectors and scalars, we take the variable matrix $X$ to be block-diagonal, where the first block has rows and columns indexed by $\mathcal{V} \times [q]$, with $X_{(i, a), (j, b)}$ representing the inner product $\langle \mathbf{b}_{i, a}, \mathbf{b}_{j, b} \rangle$. The second block is indexed by $\prod_{P \in \mathcal{P}} [q]^{\mathcal{V}(P)}$, where each diagonal entry $X_{(P, f), (P, f)}$ represents $\mu_P(f)$ and off-diagonal entries are zero. All of the constraints can then easily be written as linear constraints on entries of $X$. Also, observe that $X$ is semidefinite if and only if both blocks are. Since the second block is diagonal and all entries are nonnegative anyway, it is always semidefinite, so $X$ is semidefinite if and only if the first block is, which happens if and only if there exist vectors $\mathbf{b}_{i, a}$ for each $i \in \mathcal{V}$, $a \in [q]$, such that $X_{(i, a), (j, b)} \equiv \langle \mathbf{b}_{i, a}, \mathbf{b}_{j, b} \rangle$. Thus, we indeed have a semidefinite program.

To define an interpretation of $\tau_\txt{SDP}$ in $\tau_\Lambda$, the first step is to define the universe of the index sort in $\tau_\txt{SDP}$ (the bit sort and constraint sort universes must be defined as well, but they are much easier, so we ignore them) in terms of the universe of $\tau_\Lambda$, which is $\mathcal{V}$. As described in the previous paragraph, the universe of the index sort must represent
$$(\mathcal{V} \times [q])\ \dot{\cup}\ \prod_{P \in \mathcal{P}} [q]^{\mathcal{V}(P)}.$$
Let $\seq{P}{m}$ be the constraint types and let $k$ be the maximum arity of any constraint (recall that there are only finitely many constraint types allowed, so $k$ and $m$ are universal constants for the problem $\Lambda$, and do not depend on the instance at hand). Adding extra (ignored) variables to the second block of the matrix, we can enlarge the index set to be
$$(\mathcal{V} \times [q])\ \dot{\cup}\ \left([m] \times [q]^k \times \mathcal{V}^k\right),$$
representing an index $(P, f)$ as $(t, \seq{a}{k}, \seq{i}{k})$, where $P$ is of type $P_t$ and the first set of consecutive $i$-variables up to the arity of $P$ are assigned the corresponding $a$-values. Thus, the index sort can be constructed from $q$ disjoint copies of $\mathcal{V}$ and $mq^k$ disjoint copies of $\mathcal{V}^k$. Since $q$, $m$ and $k$ are constants, a universe like this can be defined via FO formulas using the \emph{method of finite expansions} \cite[Sec.\txt{} 2.2]{DefinableInapproximabilityJournal}.

Similarly as with \maxcut, defining the rest of the interpretation is easy but tedious. Inspecting the \LC relaxation (\ref{equLC2}), it is clear that all coefficients of the objective matrix, constraint matrices and constraint bounds can be defined via FO formulas and simple arithmetic, so the SDP can be defined via a FO interpretation.

After solving the SDP, Raghavendra's algorithm then proceeds to round the $\mathbf{b}_{i, a}$ vectors to an integral solution \cite[Theorem 5.1]{RagThesis}. As with the Goemans-Williamson algorithm, there are many symmetry-breaking steps in this process, so we must find a different way to extract the approximately optimal value.

For any $I \in \Lambda$ (that is, $I$ is an instance of CSP $\Lambda$) let $\opt(I)$ denote the maximal value $I$, and let $\sdp(I)$ denote the maximal value of the \LC relaxation of $I$ (which may be greater). For any $c \in \rr$, define
$$\gaplambda(c) := \inf_{I \in \Lambda,\ \sdp(I) = c} \opt(I).$$
As with \maxcut, we can compute the value of $\sdp(I)$ in FPC by applying Theorem \ref{thmSDPInFPC}. Using that value alone, the best approximation algorithm we can hope for would be to simply return $\gaplambda(\sdp(I))$. If the goal is just to guarantee an approximation ratio of $\alpha$ for some fixed constant $\alpha$ which is the worst-case ratio between $a$ and $\gaplambda(a)$ overall all values $a$, then we are done, for we can just return $\alpha \cdot \sdp(I)$, which is computable in FPC. However, Raghavendra's rounding algorithm has a stronger performance guarantee: that it always returns a solution of value at least $\gaplambda(\sdp(I) - \eta) - \eta$ for any fixed constant $\eta > 0$. To meet this guarantee without breaking symmetry, we instead use another one of Raghavendra's results.\thmspace

\begin{theorem}[\txt{Raghavendra \cite[Theorem 5.2]{RagThesis}}]\label{thmRagComputeGapCurve}
	For every constant $\eta > 0$ and every CSP $\Lambda$, $\gaplambda(c)$	can be computed to an additive approximation of $\eta$ in time $\exp(\exp(\textup{poly}(kq/\eta)))$, where $k$ and $q$ are constants depending only on $\Lambda$.
\end{theorem}

This algorithm approximates the infimum over all instances by computing $\opt(I)$ and $\sdp(I)$ for a set of instances $S$ of size $\exp(\exp(\txt{poly}(kq/\eta)))$. Note that $k$, $q$ and $\eta$ are constants that do not depend on $I$, so $S$ is just a fixed, finite set of instances. So, in other words, the mapping $c \mapsto \gaplambda(c)$ is essentially stored in a large, but finite, lookup table, where the value of $\gaplambda(c)$ on  an arbitrary input $c$ is approximated by looking at the greatest entry of the table below $c$. Using a (very large) disjunction over all of the entries in the table, is possible to write a FO interpretation $\Theta$ of $\tau_\qq$ in $\tau_\qq$ approximating $\gaplambda(c)$, of the form
$$\Theta(c) := \max_{\substack{I \in S\\\sdp(I) < c}} \opt(I) - \eta$$
(as described in Section \ref{secGWAlgoInFPC}, this can be translated into more primitive logical definitions of each bit in the numerator and denominator).

Composing this interpretation with the interpretations defining and solving the \LC relaxation, we have the following result, generalizing Theorem \ref{thmMaxCutInFPC}.\thmspace
\begin{theorem}\label{thmCSPApproxInFPC}
	For any CSP $\Lambda$ and any $\varepsilon > 0$, there is an FPC-definable algorithm, which, on instance $I \in \Lambda$ of SDP value $\sdp(I) = c$, returns a value of at least $\gaplambda(c) - \varepsilon$.
\end{theorem}
Note that there are two sources for the $\varepsilon$ error term, coming from the $\delta$ in Theorem \ref{thmSDPInFPC} and the $\eta$ in Theorem \ref{thmRagComputeGapCurve}.


\section{An FPC analogue of Raghavendra's result on unique games and semidefinite programming}\label{secUGAndSDPReduction}

The performance guarantee of Raghavendra's algorithm is optimal in the following sense.\thmspace
\begin{theorem}[\txt{Raghavendra \cite[Theorem 7.1]{RagThesis}}]\label{thmRagUGHardness}
	Assume the Unique Games Conjecture. For any CSP $\Lambda$, for all $\eta > 0$ and $-1 < c \leq 1$, it is NP-hard to distinguish between instances $I \in \Lambda$ with value at least $c - \eta$ from those with value at most $\gaplambda(c)$.
\end{theorem}
As a consequence, if we assume the UGC and $\PP \neq \NP$, Raghavendra's algorithm gives the best possible approximation ratio of any polynomial time algorithm. Our goal is to prove an FPC-analogue of this result, so first we must define an FPC-version of the UGC.\thmspace
\begin{conjecture}[FPC-UGC]\label{cnjFPCUGC}
	For all $\varepsilon, \delta > 0$, there exists $q$ such that there is no sentence $\phi$ of FPC such that, for all $\tauuug$ structures $\mathbb{A}$,
	\begin{enumerate}[label={(\arabic*)}]
		\item\label{itmFPCUGCCompleteness} if at least a $1 - \varepsilon$ fraction of constraints in $\mathbb{A}$ can be satisfied, then $\mathbb{A} \models \phi$, and
		\item\label{itmFPCUGCSoundness} if at most a $\delta$ fraction of constraints in $\mathbb{A}$ can be satisfied, then $\mathbb{A} \not\models \phi$.
	\end{enumerate}
\end{conjecture}

Theorem \ref{thmRagUGHardness} is proved via a series of gap-preserving reductions from \ug to $\Lambda$. We show that each of these is an FPC-reduction.

First is an elementary reduction of Khot \cite[Sec.\txt{} 2.4]{UGCSurvey} taking as input an arbitrary \ug instance and producing a new instance with approximately the same satisfiability such that the underlying graph structure is bipartite. Basically, two duplicate copies of the variable set are created, and for each constraint of the input instance between a pair of variables $(u, v)$, we have constraints in the new instance between $u$ in the first copy and $v$ in the second copy, and vice versa. Next is a series of three reductions due to Khot and Regev \cite[Lemmas 3.3, 3.4 and 3.6]{KhotRegevUGVersions} taking as input a bipartite, possibly weighted \ug instance and producing a bipartite, unweighted instance of approximately the same optimal value, with some additional useful properties in the case where the input is highly satisfiable. With the exception of Lemma 3.4, all of these reductions are ``gadget reductions" which are easily implemented as FO reductions, with FO formulas defining what constraints appear and simple arithmetical operations defining their weights.

Lemma 3.4 passes from a weighted to an unweighted instance by replacing weighted edges by multiple edges sharing the same constraint\footnote{Even though we started with an unweighted instance, Lemma 3.3 produces a weighted instance, so it is still necessary to perform this reduction. }. Inevitably, there is a slight bit of rounding error that occurs. To ensure that all vertices on the left side of the bipartite graph have the same degree (an important property used later), each vertex $x$ on the left side is assigned an arbitrary vertex $y_0(x)$ of positive weight on the right side, then weights are rounded down on all edges from $x$ besides the one to $y_0(x)$, and any extra edges are added to $y_0(x)$ in the end. Choosing $y_0(x)$ for each $x$ arbitrarily clearly breaks symmetry, preventing this reduction from being translated into an FPC-interpretation. To circumvent this problem, we may simply define $y_0(x)$ to be the variable corresponding to $x$ from Khot's reduction which duplicated the variable set (one can easily verify that the uniqueness of this choice is preserved by the reduction in Lemma 3.3 of \cite{KhotRegevUGVersions}).

Finally, we come to Raghavendra's reduction from \ug to $\Lambda$. It is defined with respect to some fixed instance $I$ of CSP $\Lambda$ with variables taking values in $[q]$. The input to the reduction is an unweighted unique games instance $\Phi$ with label set $[R]$, whose underlying graph $G$ is bipartite, with vertex set $V(G) = \mathcal{W}_\Phi \cup \mathcal{V}_\Phi$. The output is a \emph{verifier}, which is an algorithm that randomly selects a constraint of the form allowed by $\Lambda$ and returns a numeric ``payoff" if the constraint is satisfied, and zero if it is not satisfied. The verifier can therefore be viewed as an instance of $\Lambda$, where the weight of a constraint is the probability it is selected by the verifier multiplied by the payoff.

The input to the verifier is an assignment $\mathcal{F}$ of values in $[q]$ to every element of $\mathcal{V}_\Phi \times [q]^R$. The verifier performs the following steps\footnote{This is paraphrased from Raghavendra's verifier \cite[Sec.\txt{} 7.5]{RagThesis} and its main subroutine, the ``dictatorship test" \cite[Sec.\txt{} 7.3]{RagThesis}. }:
\begin{enumerate}[label={(\arabic*)}]
	\item\label{itmUGAndSDPPickSubset} Pick a constraint $P$ of $I$, uniformly at random\footnote{Raghavendra's framework \cite[Definition 2.4.2]{RagThesis} allows for constraints of CSPs to have probabilities associated with them which get multiplied by the payoffs/weights in calculating the value of an assignment of variables, in which case those probabilities should be used as a distribution for this step, rather than the uniform distribution. We have not included these probabilities because they are redundant, as they can without loss of generality be absorbed into the payoffs/weights on the constraints. }. Denote the variables in $P$ by $\mathcal{V}(P) =: \{\seq{s}{k}\}$.
	\item\label{itmUGAndSDPPickXVertex} Pick a random vertex $w \in \mathcal{W}_\Phi$.
	\item\label{itmUGAndSDPPickNeighbors} Independently pick $k$ neighbours $\seq{v}{k} \in \mathcal{V}_\Phi$ of $w$, uniformly at random (there may be duplicates).
	\item\label{itmUGAndSDPGenerateZTildeVectors} Independently generate $k$ $[q]$-valued vectors of length $R$, $\tilde{z}_{s_1}, \tilde{z}_{s_2}, \dots, \tilde{z}_{s_k}$ using a random procedure that depends only on the instance $I$ and not on $\Phi$.
	\item\label{itmUGAndSDPPermuteAndLookup} For each $i \in [k]$, permute the components of the vector $\tilde{z}_{s_i}$ by $\pi_{w, v_i}$ to obtain the vector $z_i := \pi_{w, v_i}(\tilde{z}_{s_i})$, then check the value of the variable $(v_i, z_i)$ under $\mathcal{F}$.
	\item\label{itmUGAndSDPEvaluate} If $P$ is satisfied by $k$-tuple of values obtained in step \ref{itmUGAndSDPPermuteAndLookup}, return the weight of $P$, otherwise return 0.
\end{enumerate}

One can view this verifier as an instance $f(\Phi)$ of the CSP $\Lambda$ where the variable set is $\mathcal{V}_\Phi \times [q]^R$. Similarly as discussed in Section \ref{secRagAlgo}, it is possible to define this universe as $q^R$ disjoint copies of $\mathcal{V}_\Phi$ in FPC. So, in order prove that the reduction $\Phi \mapsto f(\Phi)$ can be realized as an FPC-interpretation, all that remains is to show that the weights of each of the constraints of $f(\Phi)$ can be defined in FPC.

The weight of an arbitrary constraint $P$ of type $t$ occurring on an arbitrary tuple of variables $(\mathbf{v}, \mathbf{z}) = ((v_1, z_1), (v_2, z_2), \dots, (v_k, z_k)$ is
$$\mathbb{P}\{P \txt{ is chosen in step \ref{itmUGAndSDPPickSubset}}\} \cdot \mathbb{P}\{(\mathbf{v}, \mathbf{z}) \txt{ is queried in step \ref{itmUGAndSDPPermuteAndLookup}}\} \cdot \weight_I(P)$$
The product of the first and last of these three terms is 0 if there is no such constraint $P$ of type $t$ occurring on those variables in $I$, and
$$\frac{\weight_I(P)}{\txt{\# of constraints in $I$}}$$
otherwise (in which case $P$ is unique). This is clearly definable in FPC, so all that remains is to check that the middle term is definable in FPC as well. This can be rewritten as 
\begin{align*} 
	&\mathbb{P}\{(\mathbf{v}, \mathbf{z}) \txt{ is queried in step \ref{itmUGAndSDPPermuteAndLookup}}\}\\
	=&\prod_{i \in [k]} \Big(\mathbb{P}\{v_i \txt{ is the $i\tth$ vertex chosen in step \ref{itmUGAndSDPPickNeighbors}}\} \push\cdot \mathbb{P}\{z_i \txt{ is the $i\tth$ vector computed in step \ref{itmUGAndSDPPermuteAndLookup}}\}\Big)\\
	=&\frac{1}{\abs{\mathcal{W}_\Phi}}\sum_{w \in \mathcal{W}_\Phi} \prod_{i \in [k]} \mathbb{P}\{v_i \txt{ is the $i\tth$ vertex chosen in step \ref{itmUGAndSDPPickNeighbors}}\push\txt{given $w$ is chosen in step \ref{itmUGAndSDPPickXVertex}}\}\push\cdot \mathbb{P}\{z_i \txt{ is the $i\tth$ vector computed in step \ref{itmUGAndSDPPermuteAndLookup}}\push\txt{given $w$ is chosen in step \ref{itmUGAndSDPPickXVertex}}\}\\
	=& \frac{1}{\abs{\mathcal{W}_\Phi}}\sum_{w \in \mathcal{W}_\Phi \txt{ s.t. } \forall i \in [k],\ \{v_i, w\} \in E(G)} \frac{1}{(\deg v)^k} \push\cdot \prod_{i \in [k]}\mathbb{P}\{z_i \txt{ is the $i\tth$ vector computed in step \ref{itmUGAndSDPPermuteAndLookup}}\push\txt{given $w$ is chosen in step \ref{itmUGAndSDPPickXVertex}}\}\\
	=& \frac{1}{\abs{\mathcal{W}_\Phi}}\sum_{w \in \mathcal{W}_\Phi \txt{ s.t. } \forall i \in [k],\ \{v_i, w\} \in E(G)} \frac{1}{(\deg v)^k} \push\cdot \prod_{i \in [k]} \mathbb{P}\{\pi_{w, v_i}\inv(z_i) \txt{ is the $i\tth$ vector drawn in step \ref{itmUGAndSDPGenerateZTildeVectors}}\}.
\end{align*}

Since the procedure in step \ref{itmUGAndSDPGenerateZTildeVectors} depends only on the fixed instance $I$ and not $\Phi$, computing the probabilities in the last line can be done by lookup, e.g., by writing out all of the cases in a long FPC sentence. Hence, this entire formula boils down to simple arithmetic and case analysis, so it can be translated into an FPC-interpretation using the methods discussed in previous sections.

Composing all of the FPC-interpretations together, we have the following result, paralleling Theorem \ref{thmRagUGHardness}.\thmspace
\begin{theorem}\label{thmFPCUGHardness}
	Assume Conjecture \ref{cnjFPCUGC} (FPC-UGC). For a any CSP $\Lambda$, for all $\eta > 0$ and $-1 < c \leq 1$, there is no sentence $\phi$ of FPC such that $\phi$ is satisfied by all $\tau_\Lambda$-structures $\mathbb{A}$ with maximal value at least $c - \eta$ and $\phi$ is unsatisfied by those with maximal value at most $\gaplambda(c)$.
\end{theorem}
\begin{proof}
	Suppose there was such a $\phi$. As Raghavendra shows, the reduction from \ug to $\Lambda$ maps highly satisfiable \ug instances to $\tau_\Lambda$-structures with maximal value at least $c - \eta$, and maps highly unsatisfiable \ug instances to $\Lambda$ instances with maximal value at most $\gaplambda(c)$. As we have argued in this section, this reduction is definable as an FPC-interpretation $\Theta$. Therefore, if we ``compose" $\Theta$ with $\phi$ (replacing relation symbols in $\phi$ by their definitions as $\tauuug$-formulas according to $\Theta$, etc.), we get an FPC sentence in vocabulary $\tauuug$ violating Conjecture \ref{cnjFPCUGC}. Hence, no such $\phi$ exists.
\end{proof}

As a consequence, if we assume the FPC-UGC, there is no better FPC-definable approximation algorithm than that of Theorem \ref{thmCSPApproxInFPC}.


\chapter{Results on Unique Games}\label{chaUG}

As Chapter \ref{chaCSPs} has shown, understanding the limits to which \ug can be approximated in FPC is the key gap in our knowledge of definable inapproximability of CSPs. Thus, in this chapter, we examine \ug in detail. The fundamental question we ask is, given some fixed integer $k$, to what extent can a sentence $\phi$ of FPC where $\mu(\phi) = k$ separate instances of different optimal values? To build intuition, we begin by considering the simple cases where $k \in \{1, 2, 3\}$. We then present a novel CFI-construction proving that there is no FPC-interpretation giving the exact optimal value of a \UG(4) instance. This construction is then generalized to prove the main result of this chapter (Theorem \ref{thmUGLowGapMain}), that it is impossible to approximate the optimal value of a \UG($q$) instance to within any constant factor $\alpha$ in FPC (where $q = \txt{poly}(\frac{1}{\alpha})$).

While none of the constructions in this chapter yield stronger lower bounds on the approximability of \ug than what are known for polynomial time computation, they are still valuable for two main reasons. First, the lower bounds do not rely on the assumption that $\PP \neq \NP$, so the results are truly novel. Second, the constructions themselves are qualitatively quite different from existing \ug constructions in that they exploit a particular weakness of FPC-definable algorithms: the inability to solve systems of linear equations. As such, they provide a new set of tools with which to attack Conjecture \ref{cnjFPCUGC}.


\section{The label-lifted instance}\label{secLiftedGraph}

Recall from Sections \ref{subPebblingGameDefinition} and \ref{secDefinableInapproximabilityReview} that, to establish FPC inapproximability results for a given problem, it suffices to produce, for any integer $k$, two instances $\mathbb{A}_k$ and $\mathbb{B}_k$ (for \ug, these are $\tauuug$-structures, using the unweighted encoding) with very different optimal values such that Duplicator wins the $k$-pebble bijective game played on $\mathbb{A}_k$ and $\mathbb{B}_k$.

For $k = 1$, this is fairly trivial. Just let $\mathbb{A}_1$ and $\mathbb{B}_1$ be \ug instances on the same number of variables, whose underlying graphs are simple, such that $\mathbb{A}_1$ is completely satisfiable and $\mathbb{B}_1$ is highly unsatisfiable (such instances $\mathbb{B}_1$ are easy to construct; an explicit construction is given in Appendix \ref{appUnsatisfiableUGConstruction}). Since the two structures have the same size, Duplicator is always able to give a bijection between their universes. No matter what bijection Duplicator chooses, Spoiler can never win, since all relations in $\tauuug$ have arity 2 and there is only one pebble (and $\mathbb{B}_1$ has no self-loops).

For $k = 2$, we must be more clever, since we now have to ensure that, when there is already one pebble pair on the board, Duplicator's bijection preserves all of the edge labels incident to the pebbled vertices. To this end, we define an operator $\mathcal{G}$ on \gug instances (defined at the end of Section \ref{subUGBackground}), similar to the $G$ operator used by Atserias and Dawar \cite[Sec.\txt{} 3.2]{DefinableInapproximabilityJournal}, and also implicitly used by Atserias, Bulatov and Dawar \cite[Sec.\txt{} 3]{SolvingEquationsNotInFPC}.

Given a \gug instance $U$ with group $A$ and variable set
$$\{x_v \suchthat v \in V\},$$
$\mathcal{G}(U)$ is a \gug instance with group $A$ and variable set
$$\{x_v^g \suchthat v \in V,\ g \in A\}.$$
For every equation
$$x_{v_1} - x_{v_2} = z$$
in the constraint set of $U$ and every $g_1, g_2 \in A$, we have the equation
$$(x_{v_1}^{g_1} - g_1) - (x_{v_2}^{g_2} - g_2) = z$$
in the constraint set of $\mathcal{G}(U)$. We call $\mathcal{G}(U)$ the \emph{label-lifted instance\footnote{This construction is similar to the \emph{label-extended graph} of a \ug instance (see, for example, \cite{LabelExtendedGraphExample, UGAT}), but it is not the same thing. The label extended-graph is obtained by taking all of the edges with identity constraints in the label-lifted instance. } of $U$}.

The hope is that it is easier for Duplicator to win the $k$-pebble bijective game on $\mathcal{G}(U_1)$ and $\mathcal{G}(U_2)$ than on the original pair $U_1$ and $U_2$, while at the same time, applying $\mathcal{G}$ does not change how satisfiable an instance is. Formally, for any $s \in [0, 1]$, we say that an (unweighted) \ug instance $U$ is \emph{$s$-satisfiable} if there is some assignment of variables satisfying at least an $s$-fraction of the constraints of $U$. We say that the \emph{satisfiability} of $U$ is the maximum $s$ such that $U$ is $s$-satisfiable.\thmspace
\begin{lemma}\label{lemGSameSatisfiability}
	\lemGSameSatisfiability
\end{lemma}
\begin{proof}[Proof sketch.]
	If $x_v$ is an assignment\footnote{We sometimes use a symbol like $x_v$ or $x_v^g$ to denote a specific variable, and sometimes to denote the value assigned to that variable. When $v$ is unspecified, as it is here, we mean a function assigning a value to each variable. Throughout this chapter, it should be clear from context which of the three meanings we intend. } satisfying at least an $s$ fraction of the constraints of $U$, then it is not too hard to see that the assignment $x_v^g := x_v + g$ satisfies at least an $s$ fraction of the constraints of $\mathcal{G}(U)$. For the other direction, given an assignment $x_v^g$ satisfying an $s$-fraction of the constraints of $\mathcal{G}(U)$, we argue that there exists some mapping $f: V \to A$ such that an $s$-fraction of the constraints between the variables $\{x_v^{f(v)} \suchthat v \in V\}$ are satisfied. It is then shown that the assignment $x_v := x_v^{f(v)} - f(v)$ satisfies at least an $s$-fraction of the constraints of $U$. See Appendix \ref{appGSameSatisfiability} for the details.
\end{proof}


\section{The case of $k = 2$}\label{secKEquals2}
We are now able to prove the FPC-UGC (Conjecture \ref{cnjFPCUGC}) in the special case where $\mu(\phi) = 2$. In fact, the result is slightly stronger because it holds even for $\varepsilon = 0$.\thmspace
\begin{theorem}\label{thmKEqualsTwoInapproximability}
	For all $\delta > 0$, there exists an integer $q$ such that there is no sentence $\phi$ of FPC such that $\mu(\phi) = 2$ and, for all $\tauuug$ structures $\mathbb{A}$,
	\begin{enumerate}[label={(\arabic*)}]
		\item\label{itmKEquals2Completeness} if $\mathbb{A}$ is completely satisfiable, then $\mathbb{A} \models \phi$, and
		\item\label{itmKEquals2Soundness} if $\mathbb{A}$ is not $\delta$-satisfiable, then $\mathbb{A} \not\models \phi$.
	\end{enumerate}
\end{theorem}
\begin{proof}
	Let $\delta$ be given, and suppose toward a contradiction that there did exist some sentence $\phi$ satisfying \ref{itmKEquals2Completeness} and \ref{itmKEquals2Soundness}. Let $U_2$ be any \gug defined over an underlying graph $G$ that is simple such that $U_2$ is not $\delta$-satisfiable (see Appendix \ref{appUnsatisfiableUGConstruction}). We then define $U_1$ to be the \gug instance obtained by turning all of the constraints in $U_2$ into identity constraints (so $U_1$ is completely satisfiable). We claim that $\mathcal{G}(U_1) \equiv_{C^2} \mathcal{G}(U_2)$.
	
	To prove this, let the variable sets of $U_1$ and $U_2$ be
	$$\{x_v \suchthat v \in V\},$$
	so that the variable sets of $\mathcal{G}(U_1)$ and $\mathcal{G}(U_2)$ are
	$$X = \{x_v^g \suchthat v \in V,\ g \in A\},$$
	as in the definition of the label-lifted instance. Duplicator's strategy in the 2-pebble bijective game is to always give a bijection $f: X \to X$ (from the universe of $\mathcal{G}(U_1)$ to the universe of $\mathcal{G}(U_2)$) with the following property:
	\begin{equation}\label{equKEquals2PreserveProp}
		\txt{For all $v \in V$, there exists $g^*(v) \in A$ such that $f(x_v^g) \equiv x_v^{g + g^*(v)}$.}
	\end{equation}
	So, in any given round, Duplicator's bijection is completely determined by a map $g^*: V \to A$.
	
	If there are no pebbles on the board when duplicator is giving a bijection, then Duplicator can choose any $g^*$. Otherwise, suppose that one pebble pair is on $x_{v_0}^{g_1}$ in the universe of $U_1$ and the corresponding pebble is on $x_{v_0}^{g_2}$ in the universe of $U_2$ (it must be the same $v_0$ for both elements, since we may assume inductively all of Duplicator's previous bijections satisfied (\ref{equKEquals2PreserveProp})). Then, for any $v \in V$, we define 
	$$g^*(v) := \threecases{\txt{if } v = v_0}{g_2 - g_1}{\txt{if } x_v - x_{v_0} = g_3 \txt{ is an equation in } U_2}{g_2 - g_1 - g_3}{\txt{otherwise}}{\txt{anything}}$$
	Note that the middle case is uniquely defined and disjoint from the first case since $G$ has no multiple edges or self-loops. The bijection $f: X \to X$ determined by $g^*$ is valid because it respects the pebble pair which is already placed:
	$$f(x_{v_0}^{g_1}) = x_{v_0}^{g_1 + g^*(v_0)} = x_{v_0}^{g_1 + (g_2 - g_1)} = x_{v_0}^{g_2}.$$
	
	Suppose that Spoiler places the second pair of pebbles on some arbitrary $(x_v^g, f(x_v^g))$. The only way Spoiler could win at this step in the game is if there was some equation between $x_v^g$ and $x_{v_0}^{g_1}$ in $\mathcal{G}(U_1)$ with no matching equation between $f(x_v^g)$ and $x_{v_0}^{g_2}$, or vice versa. Such an equation can only exist in either graph if $v$ and $v_0$ are neighbours, in which case $g^*(v)$ is defined according to the middle case, implying that, for any $g_3 \in A$,
	\begin{equation*}
		x_v - x_{v_0} = g_3 \txt{ is an equation in } U_2 \iff g^*(v) = g_2 - g_1 - g_3.
	\end{equation*}
	Therefore, for any arbitrary group element $g_4 \in A$,
	\begin{align*} 
		&&x_v^g - x_{v_0}^{g_1} = g_4 &\txt{ is an equation in } \mathcal{G}(U_1)\\
		\iff&& (x_v^g - g) - (x_{v_0}^{g_1} - g_1) = g_4 - g + g_1 &\txt{ is an equation in } \mathcal{G}(U_1)\\
		\iff&& x_v - x_{v_0} = g_4 - g + g_1 &\txt{ is an equation in } U_1\\
		\iff&& g_4 - g + g_1 = 0\\
		\iff&& g^*(v) = g_2 - g_1 - (g_4 - g - g^*(v) + g_2)\\
		\iff&& x_v - x_{v_0} = (g_4 - g - g^*(v) + g_2) &\txt{ is an equation in } U_2\\
		\iff&& (x_v^{g + g^*(v)} - g - g^*(v)) - (x_{v_0}^{g_2} - g_2) \\&&= g_4 - g - g^*(v) + g_2 &\txt{ is an equation in } \mathcal{G}(U_2)\\
		\iff&&x_v^{g + g^*(v)} - x_{v_0}^{g_2} = g_4 &\txt{ is an equation in } \mathcal{G}(U_2)\\
		\iff&&f(x_v^g) - x_{v_0}^{g_2} = g_4 &\txt{ is an equation in } \mathcal{G}(U_2),
	\end{align*}
	so Spoiler is unable to reveal a difference between the two structures. Since Spoiler can never win, Duplicator wins, and hence $\mathcal{G}(U_1) \equiv_{C^2} \mathcal{G}(U_2)$. However, by Lemma \ref{lemGSameSatisfiability}, $\mathcal{G}(U_1)$ must satisfy $\phi$ since it is completely satisfiable, while $\mathcal{G}(U_2)$ must \emph{not} satisfy $\phi$ since it is not $\delta$-satisfiable. As $\mu(\phi) = 2$, this contradicts the fact that $\mathcal{G}(U_1) \equiv_{C^2} \mathcal{G}(U_2)$; hence, no such sentence $\phi$ exists.
\end{proof}

This kind of argument is used several more times throughout this chapter, with slight variation. The meta-theorem is that, \emph{for every constraint in $U$ which is satisfied by $g^*$, the corresponding function $f$ is a partial isomorphism over all of the corresponding constraints in $\mathcal{G}(U)$.} We do not state this as a formal theorem because there are some technicalities involved when we ascribe a different meaning to the word ``satisfied" in later sections.


\section{FPC-inexpressibility of solving \ug exactly}\label{secUGExact}

Starting at $k = 3$, we run into trouble in proving FPC-inapproximability bounds, due to the following fact:\thmspace
\begin{proposition}\label{proUGCompletelySatisfiableInFPC}
	For any positive integer $q$, there is a sentence $\phi$ of LFP, where $\mu(\phi) = 3$, expressing the property that a $\textup{\UG}(q)$ instance (encoded as a $\tauuug$-structure) is completely satisfiable.
\end{proposition}
\begin{proof}
	For each fixed label $i \in [q]$, we define $q$ unary relations $U_{i, 1}, U_{i, 2}, \dots, U_{i, q}$, parameterized by a free variable $x$, by simultaneous induction:
	\begin{align*} 
	U_{i, i}(y) &\impliedby (x = y)\\
	U_{i, j}(y) &\impliedby \exists x \bigvee_{\pi: [q] \to [q]} \left(P_\pi(y, x) \wedge U_{i, \pi(j)}(x)\right)
	\end{align*}
	Note that the ``$\exists x$" term creates a new variable $x$, different from the $x$ in the first line (this is done purely in an effort to reduce the total number of variables). The meaning of $U_{i, j}(y)$ is that, \emph{given $x$ has label $i$, it is implied by the constraints that $y$ has label $j$}. Thus, $U_{i, i}(x)$ is defined to be true, and whenever a constraint $P_\pi$ holds on a pair of elements $(y, x)$ and we know what the label of $x$ must be, we inductively derive what the label of $y$ must be. We claim that the following sentence expresses the property that a $\textup{\UG}(q)$ instance is completely satisfiable:
	$$\phi \equiv \forall x \bigvee_{i \in [q]} \bigwedge_{\substack{j \in [q]\\j \neq i}} \neg U_{i, j}(x)$$
	If $f$ is a satisfying assignment, then picking $i = f(x)$ must satisfy the formula. Conversely, if the formula is satisfied, one can obtain a satisfying assignment by picking one $x$ from each connected component of the underlying graph and one satisfying witness $i$, then assigning labels to every $y$ in that component by taking the unique $j$ such that $U_{i, j}(y)$ holds. (It is not too hard to see that the component being connected implies $j$ exists, and the formula being satisfied implies that $j$ is unique).
	
	Using the \emph{Bekic principle} \citetext{\citealp[Lemma 1.4.2]{RudimentsOfMuCalculus}; \citealp[Lemma 10.9]{FMT}}, the simultaneous inductions can be nested within each other in a way that reuses variable names, resulting in LFP formulas for each of the $q$ relations, still using only 2 variables ($x$ and $y$). Thus, $\phi$ can indeed be written as an LFP sentence of only 2 variables. This can then be translated into a $C^3$ sentence using \cite[Lemmas 1.28 and 1.29]{InfinitaryTranslation}, so $\mu(\phi) = 3$. (The resulting sentence has one more variable due to Lemma 1.28. Each unary relation in $\phi$ requires one variable to be locally renamed.)
\end{proof}

Since LFP is a restriction of FPC, it is therefore impossible to prove a result along the lines of Theorem \ref{thmKEqualsTwoInapproximability} if $\mu(\phi) \geq 3$. However, even though it is possible to determine in FPC whether an instance is completely satisfiable, we now prove that, in the case where the input is \emph{not} completely satisfiable, it is impossible to determine in FPC what the optimal value actually is. This result is, in a way, superseded by that of Section \ref{secUGLowGap}, but the construction is interesting in its own right, and serves as a motivation for the more complicated construction in the latter proof.

Fix a positive integer $k$. We exhibit a pair of $\tauuug$-structures, $\mathbb{A}_k$ and $\mathbb{B}_k$, such that $\mathbb{A}_k$ has a strictly greater optimal value than $\mathbb{B}_k$, but $\mathbb{A}_k \equiv_{C^k} \mathbb{B}_k$. Let $H = H_k$ be a simple graph satisfying the following four properties:
\begin{enumerate}[label={(\arabic*)}]
	\item\label{itmUGExactConnected} $H$ is connected.
	\item\label{itmUGExactReg} $H$ is 3-regular.
	\item\label{itmUGExactBipartite} $H$ is bipartite.
	\item\label{itmUGExactRobberWin} The robber player wins the $(k - 1)$-cop edge-robber game\footnote{There is a vast literature on so called \emph{Cops and Robbers} games, in which a team of cops is trying to catch a robber on a graph. See \cite{CopsAndRobbersSurvey} for a survey. The author is unaware if this particular variant has already been studied.}, as defined below.
\end{enumerate}

\begin{addmargin}[1em]{2em}
	\textbf{The $k$-cop edge-robber game played on a connected graph $H$:} There are $k$ \emph{cops}, controlled by the \emph{cop player}, which may be placed on vertices of $H$; and 1 \emph{robber}, controlled by the \emph{robber player}, which may be placed on edges of $H$. The robber is initially placed on an arbitrary edge in $H$, and the cops are initially not placed anywhere. The cop player and robber player take turns, starting with the robber player. On the robber player's turn, the robber player may move the robber along any path of edges in $H$ that is unobstructed by a cop. On the cop player's turn, the cop player may move one cop to any vertex of the graph without restriction. The cop player wins if the cops can \emph{capture} the robber on an edge $\{u, v\}$, meaning surround it with cops on both vertices $u$ and $v$. The robber player wins by infinitely evading capture.
\end{addmargin}

It is not too hard to show that such graphs exist; see Appendix \ref{appCopsAndRobbers} for details. Similarly as done by Atserias, Bulatov and Dawar \cite[Lemma 3]{SolvingEquationsNotInFPC}, we use the robber player's winning strategy on the $(k - 1)$-cop edge-robber game played on $H$ as part of Duplicator's strategy in the $k$-pebble bijective game on $\mathbb{A}_k$ and $\mathbb{B}_k$.

The \ug instances we construct are \gug instances using the Klein four-group as $A$. For convenience, we denote the underlying set of $A$ by $\{e, a, b, c\}$ where $e$ is the identity element. Then group addition is completely defined by the rules that
$$a + a = b + b = c + c = e,$$
and adding any two distinct elements from $\{a, b, c\}$ yields the third element.

It is well known \cite{Konig} that any $d$-regular bipartite graph always has a perfect matching, and thus by induction, that the edge set decomposes into a disjoint union of $d$ perfect matchings. Choose such a decomposition for $H$ and assign each of the 3 matchings to a distinct element of the set $\{a, b, c\}$. Let $m: E(H) \to A$ be the map the sends each edge to the element of $A$ assigned to the matching containing that edge. In other words, $m$ labels each edge with either $a$, $b$ or $c$ such that every vertex is incident to three edges that each have a different label.

Define a \gug instance $U_1$ with group $A$ and variable set
$$\{x_v \suchthat v \in V(H)\}.$$
For every edge $\{v_1, v_2\} \in E(H)$ we have two constraints: $x_{v_1} + x_{v_2} = e$ and $x_{v_1} + x_{v_2} = m(\{v_1, v_2\})$. From $U_1$, define another instance $U_2$ by choosing an edge $\{v_1^*, v_2^*\}$ such that $m(\{v_1^*, v_2^*\}) = a$, and redefining the constraints involving that edge to be $x_{v_1} + x_{v_2} = b$ and $x_{v_1} + x_{v_2} = c$.

As in the proof of Theorem \ref{thmKEqualsTwoInapproximability}, our pair of indistinguishable structures are $\mathbb{A}_k := \mathcal{G}(U_1)$ and $\mathbb{B}_k := \mathcal{G}(U_2)$.\thmspace

\begin{lemma}\label{lemUGExactCompleteness}
	The satisfiability of $U_1$ (and thus of $\mathcal{G}(U_1)$) is $\frac12$.
\end{lemma}
\begin{proof}
	Clearly, no more than $\frac12$ of the constraints of $U_1$ can be satisfied, since the constraints come in inconsistent pairs: if one is satisfied, the other must not be satisfied. The assignment $x_v := e$ attains this bound of $\frac12$ by satisfying the constraint $x_{v_1} + x_{v_2} = e$ in each pair. The claim about $\mathcal{G}(U_1)$ follows from Lemma \ref{lemGSameSatisfiability}.\thmspace
\end{proof}

\begin{lemma}\label{lemUGExactSoundness}
	\lemUGExactSoundness
\end{lemma}
\begin{proof}[Proof sketch]
	This follows from a parity argument, considering how the number of satisfied constraints changes with respect to individual changes in an assignment of variables. See Appendix \ref{appParityArgument} for the details.\thmspace
\end{proof}

\begin{lemma}\label{lemUGExactIndistinguishable}
	$\mathcal{G}(U_1) \equiv_{C^k} \mathcal{G}(U_2)$.
\end{lemma}
\begin{proof}
	While playing the $k$-pebble bijective game on $\mathcal{G}(U_1)$ and $\mathcal{G}(U_2)$, Duplicator simultaneously plays the $(k - 1)$-cop edge-robber game on graph $H$, making use of the robber player's winning strategy to hide the edge where $\mathcal{G}(U_1)$ and $\mathcal{G}(U_2)$ differ. Initially, this edge is $r_0 := \{v^*_1, v^*_2\}$.
	
	Every time Duplicator needs to give a bijection between the two universes, Duplicator first updates the positions of the cops on $H$ to be
	$$\{v \in V(H) \suchthat \txt{there is a pebble pair on $(x_v^{g_1}, x_v^{g_2})$ for some $g_1, g_2 \in A$}\}.$$
	Note that, since the previous round, Spoiler has placed one pebble pair down and picked up another, so there are at most $k - 1$ cops, at most one of which has moved from the previous turn. This constitutes a valid move of the cop player, so there must exist a winning response from the robber player. So suppose that, on the $i\tth$ round, this winning strategy dictates that the robber be moved from edge $r_{i - 1}$ to edge $r_{i}$, through a path of vertices $p_0, \seq{p}{\ell}, p_{\ell + 1}$, where $r_{i - 1} = \{p_0, p_1\}$ and $r_i = \{p_\ell, p_{\ell + 1}\}$ (so if the robber does not move, then $\ell = 0$). For each $j \in [\ell]$, let $e_j$ be the unique edge incident to $p_j$ which is not part of the path, as shown in Figure \ref{figUGExactDiagram2}.\thmspace
	
	\ipns{.4}{UGExactDiagram2}{The path of the robber.}{\label{figUGExactDiagram2}The path in $H$ from the robber's old location at edge $r_{i - 1}$ to its new location at edge $r_i$.}
	
	As in the proof of Theorem \ref{thmKEqualsTwoInapproximability}, on round $i$ of the $k$-pebble bijective game, Duplicator gives a bijection $f_i$ of the form
	$$f_i(x_v^g) := x_v^{g + g^*(i, v)}.$$
	Define $g^*(0, v) := e$ for all $v$. On every round $i$ (the $i\tth$ time Duplicator is giving a bijection), Duplicator's strategy is to set 
	$$g^*(i, v) := \twocases{\txt{if } v = p_j \in \{\seq{p}{\ell}\}}{g^*(i - 1, v) + m(e_j)}{\txt{otherwise}}{g^*(i - 1, v)},$$
	where $p_j$ and $e_j$ are determined by the robber's move on round $i$ as described in the previous paragraph. Note that the only vertices for which the bijection is different from the bijection in the previous round are those involving the vertices in $\{\seq{p}{\ell}\}$. Since the robber's move is valid, none of these vertices are occupied by a cop, and hence none of these variables are pebbled. Thus, the bijection $f_i$ is valid (assuming inductively that the bijection from the previous round $f_{i - 1}$ was valid).
	
	All that remains to prove is that Spoiler cannot expose a difference between $\mathcal{G}(U_1)$ and $\mathcal{G}(U_2)$. This follows from the following stronger claim, which we prove by induction on $i$:\thmspace
	\begin{lemma}\label{lemUGExactClaim}
		\lemUGExactClaim
	\end{lemma}
	To see how this implies that Spoiler never wins, suppose to the contrary that there was a constraint between two pebbled elements $x_{v_1}^{g_1}$ and $x_{v_2}^{g_2}$ in one structure with no matching constraint between the corresponding elements in the other structure. Lemma \ref{lemUGExactClaim} says that $f_i$ is a partial isomorphism everywhere except around the edge $r_i$, so the only way this can happen is if $\{v_1, v_2\} = r_i$. But this means that the cops have trapped the robber, which is a contradiction.
\end{proof}

The full proof of Lemma \ref{lemUGExactClaim} is rather tedious, and hence is relegated to Appendix \ref{appUGExactClaim}. Instead, let us work through a simple example where $H := K_4$, the complete graph on four vertices (this graph is not bipartite, but its edge set does decompose into a disjoint union of 3 perfect matchings, which is all we actually need). The robber player has a winning strategy in the 2-cop edge-robber game on $H$ by always moving to an edge between 2 vertices which do not have cops on them. Let us see how this strategy can be used by Duplicator to survive the first 3 rounds of the 3-pebble bijective game on $\mathcal{G}(U_1)$ and $\mathcal{G}(U_2)$, where $U_1$ and $U_2$ are the \gug instances depicted in Figure \ref{figK4Example}.

\ipnc{1.3}{K4Example2}{\label{figK4Example}One possible construction for the \gug instances $U_1$ and $U_2$ using the graph $H := K_4$, where $\{v^*_1, v^*_2\} := \{v_3, v_4\}$.}

Initially, the robber is on the edge $r_0 = \{v_3, v_4\} \in E(H)$; this is the only edge where $U_1$ differs from $U_2$, and hence the only place where Spoiler could potentially expose a difference between $\mathcal{G}(U_1)$ and $\mathcal{G}(U_2)$. As Spoiler is about to place its first pebble, there is no reason for the robber to move, so Duplicator gives the identity bijection, corresponding to $g^*(1, v) \equiv e$. Suppose Spoiler places the first pebble pair on $x_{v_1}^{g_1}$ in both instances for some $g_1 \in A$. Accordingly, Duplicator places a cop on $v_1$. As there is still no movement required from the robber, Duplicator again gives the identity bijection. Now suppose Spoiler places the second pebble pair on $x_{v_4}^{g_4}$ in both instances for some $g_4 \in A$. After Duplicator places a corresponding cop on $v_4$, the winning strategy for the robber player requires the robber to move through $v_3$ to $\{v_2, v_3\}$, as illustrated in Figure \ref{figK4CopsAndRobbers}.

\ipnc{1.2}{K4CopsAndRobbers}{\label{figK4CopsAndRobbers}The first necessary move of the robber.}

Therefore, according to Duplicator's strategy, after Spoiler picks up the third pebble pair, Duplicator updates
$$g^*(3, v_3) := m(v_1, v_3) = b.$$
Why does Lemma \ref{lemUGExactClaim} still hold? Let us verify that it holds for a particular pair of variables, $x_{v_1}^{e}$ and $x_{v_3}^{e}$. The main idea is that permuting the variables involving $v_3$ by adding $b$ does not change the set of equations between $v_1$ variables and $v_3$ variables, so if Lemma \ref{lemUGExactClaim} held in round 2, it will still hold in round 3. We initially have the equations
$$x_{v_1}^{e} + x_{v_3}^{e} = e \txt{ and } x_{v_1}^{e} + x_{v_3}^{e} = b$$
in $\mathcal{G}(U_1)$, and correspondingly,
$$\equalto{x_{v_1}^{e}}{f_2(x_{v_1}^{e})} + \equalto{x_{v_3}^{e}}{f_2(x_{v_3}^{e})} = e \txt{ and } \equalto{x_{v_1}^{e}}{f_2(x_{v_1}^{e})} + \equalto{x_{v_3}^{e}}{f_2(x_{v_3}^{e})} = b$$
in $\mathcal{G}(U_2)$. After updating $g^*$, $f_3(x_{v_3}^e) = x_{v_3}^b$, in which case the pair of equations between the $\mathcal{G}(U_2)$ variables becomes
$$x_{v_1}^{e} + x_{v_3}^{b} = b \txt{ and } x_{v_1}^{e} + x_{v_3}^{b} = b + b,$$
which is the same set of equations as between $x_{v_1}^e$ and $x_{v_3}^e$ in $\mathcal{G}(U_1)$ (since $b + b = e$), so Lemma \ref{lemUGExactClaim} is still satisfied for $x_{v_1}^{e}$ and $x_{v_3}^{e}$. The other cases follow by similar reasoning. Crucially, Lemma \ref{lemUGExactClaim} tells us that the bijection $f_3$ now preserves the constraints involving the bottom edge $\{v_3, v_4\}$ (one can check that $f_2$ does not have this property), so even if the third pebble pair is placed on variables involving $v_3$, Spoiler still does not win.

Putting together Lemmas \ref{lemUGExactCompleteness}, \ref{lemUGExactSoundness} and \ref{lemUGExactIndistinguishable}, we have the following result.\thmspace
\begin{theorem}\label{thmUGExactMain}
	For any $q \geq 4$, there is no sentence of FPC expressing the property that at least half of the constraints of a $\tauuug$ structure are satisfiable (and hence no FPC-definable algorithm returning the exact satisfiability). This holds even when restricting to \gug instances using the Klein four-group defined over an underlying graph that is 6-regular and bipartite.
\end{theorem}


\section{FPC-inapproximability of \ug}\label{secUGLowGap}

We now generalize the construction from Section \ref{secUGExact} to show that \ug is inapproximable to within any constant factor in FPC. Specifically, we show a $(\frac{1}{2^{\ell}}, \frac{\alpha}{2^{\ell}})$ inapproximability gap for $\UG(2^m)$ for arbitrarily small $\alpha$, where $m$ and $\ell$ are both $O(\log(\frac{1}{\alpha}))$.

Let us begin with a high-level overview of how the construction presented in this section compares with that of Section \ref{secUGExact}. In that construction, the instances $U_1$ and $U_2$ use the Klein four-group, which is the additive part of $\ff_2^2$, the 2-dimensional vector space over the finite field with 2 elements. Each pair of constraints (henceforth \emph{bundle} of constraints) between each pair of vertices gives Duplicator a choice when defining the bijection $f$ between the universes of $\mathcal{G}(U_1)$ and $\mathcal{G}(U_2)$: given the value of $g^*(v_1)$ there are always 2 choices of $g^*(v_2)$ which satisfy at least one of the constraints of the bundle between $v_1$ and $v_2$, which is useful, since satisfying a bundle makes $f$ \emph{locally consistent} with the corresponding relations of the two structures (meaning that it preserves constraints involving $v_1$ and $v_2$). Furthermore, these choices can be concatenated so that there are 4 ways to satisfy at least one constraint from each bundle along a path of length 2. Since $\ff_2^2$ only has 4 elements, this means that any path of length 2 can be made locally consistent given the values of its two endpoints, i.e., Duplicator can always find a value for $g^*(u)$ satisfying the bundles between $u$ and $v_1$, and between $u$ and $v_2$, given arbitrary values for $g^*(v_1)$ and $g^*(v_2)$. The construction presented in this section replaces $\ff_2^2$ with $\fm$, with a bundle of $2^\ell$ constraints between every pair of vertices, for suitably large integers $0 < \ell < m$. A key difference is that it is no longer possible to make an arbitrary path of length 2 locally consistent. However, we are able to show that there exists some $r$ such that it is possible to make any path of length $r$ locally consistent (see Lemma \ref{lemNUGLGRadiusProperty}). This allows Duplicator to win as long as the base graph $H$ has suitably high girth.

Formally, we begin by fixing constants $\varepsilon, \gamma \in (0, \frac12)$, e.g., they might as well just be defined as $\frac14$. Let $\alpha \in (0, 1]$ be given; the goal is to show that there is no FPC-definable $\alpha$-approximation algorithm for \ug. First, choose a positive integer $d$ (the degree of every vertex in $H$) large enough so that
\begin{align} 
	d &\geq \frac{16}{\alpha^2} \left(\ln(d) + 2 + \ln(2) - \ln(\varepsilon)\right)\label{equNUGLGParameterDGeqComplicated},\\
	d &> \frac{4}{(1 - 2\gamma)\alpha}. \label{equNUGLGParameterDStrictlyGreater}\\
\intertext{(Note that (\ref{equNUGLGParameterDStrictlyGreater}) automatically follows from (\ref{equNUGLGParameterDGeqComplicated}) for $\gamma := \frac14$.) Next, define the integers $\ell$ and $m$ to be}
	\ell &:= \left\lceil \log_2(d) + 2\log_2(e) \right\rceil \label{equNUGLGParameterLDef},\\
\intertext{(where $e = 2.718\dots$ is the base of the natural logarithm)}
	m &:= \left\lceil \ell - \log_2\left(\left(\frac12 - \gamma\right)\alpha - \frac{2}{d}\right) \right\rceil \label{equNUGLGParameterMDef}.\\
\intertext{Note that the logarithm in (\ref{equNUGLGParameterMDef}) is well-defined because of (\ref{equNUGLGParameterDStrictlyGreater}) and the fact that $\gamma < \frac12$. Finally, define the integer $r$ to be}
	r &:= \left\lceil m \ln(2) - \ln(\gamma)\right\rceil \label{equNUGLGParameterRGeq}.
\end{align}

For any $k$, let $\widetilde{H} = \widetilde{H}_k$ be any\footnote{Regular graphs of arbitrarily high girth and degree are known to exist; see Lazebnik, Ustimenko and Woldar \cite{DenseGraphsHighGirth}, for example. } $d$-regular simple graph of girth at least $(k + 1)^2 r$. For every edge $\{v_1, v_2\} \in E(\widetilde{H})$, independently choose a uniformly random vector $b(v_1, v_2) = b(v_2, v_1) \in \fm$ and a uniformly random $\ell$-dimensional subspace\footnote{What this means is, randomly choose a set of $\ell$ linearly independent vectors and take the span. Choose the first vector uniformly at random from $\fm \setminus \{0\}$, then choose each subsequent vector uniformly at random from the subset of $\fm$ which is not in the span of the previously chosen vectors. } $Z(v_1, v_2) = Z(v_2, v_1) \subseteq \fm$. Say that an edge $e \in E(\widetilde{H})$ is \emph{good} if, for all paths $v_0, \seq{v}{r}$ of length $r$ passing through $e$, the set
$$\bigcup_{i \in [r]} Z(v_{i - 1}, v_i)$$
spans $\fm$. Edges of $\widetilde{H}$ which are not good edges are called \emph{bad} edges.\thmspace

\begin{lemma}\label{lemNUGLGMostEdgesGood}
	\lemNUGLGMostEdgesGood
\end{lemma}
\begin{proof}
	See Appendix \ref{appNUGLGMostEdgesGood}.
\end{proof}

Let $H = H_k$ be the graph with vertex set $V(H) := V(\widetilde{H})$ and edge set
$$E(H) := \{e \in E(\widetilde{H}) \suchthat \txt{$e$ is a good edge}\}.$$

We define Group Unique Games instances $U_1$, $\widetilde{U}_1$, $U_2$ and $\widetilde{U}_2$ using the additive group structure on $\fm$. The variable sets of all four instances are
$$\{x_v \suchthat v \in V(H)\}.$$
For every edge $\{v_1, v_2\} \in E(\widetilde{H})$, $\widetilde{U}_1$ and $\widetilde{U}_2$ have $2^\ell$ constraints between the corresponding variables. In $\widetilde{U}_1$, the constraints are
$$\{x_{v_1} - x_{v_2} = z \suchthat z \in Z(v_1, v_2)\},$$
whereas in $\widetilde{U}_2$, the constraints are
$$\{x_{v_1} - x_{v_2} = z + b(v_1, v_2) \suchthat z \in Z(v_1, v_2)\}.$$
Finally, $U_1$ and $U_2$ are obtained from $\widetilde{U}_1$ and $\widetilde{U}_2$ by removing all constraints on pairs of variables corresponding to bad edges, i.e., with constraints defined in the exact same way as $\widetilde{U}_1$ and $\widetilde{U}_2$, but only for edges $\{v_1, v_2\} \in E(H)$.\thmspace

\begin{lemma}\label{lemNUGLGCompleteness}
	The satisfiability of $U_1$ (and thus of $\mathcal{G}(U_1)$) is $\frac{1}{2^\ell}$.
\end{lemma}
\begin{proof}
	As in the proof of Lemma \ref{lemUGExactCompleteness}, at most one constraint can be satisfied from each bundle, so the total satisfiability is at most $\frac{1}{2^\ell}$. The assignment $x_v^g := g$ attains this bound by satisfying the $z = 0$ constraint in each bundle (every subspace $Z(v_1, v_2)$ must contain $z = 0$).\thmspace
\end{proof}

\begin{lemma}\label{lemNUGLGSoundnessU2}
	\lemNUGLGSoundnessUTwo
\end{lemma}
\begin{proof}[Proof sketch]
	The main idea is that, since very few edges are bad edges by Lemma \ref{lemNUGLGMostEdgesGood}, the satisfiability of $\widetilde{U}_2$ closely approximates the satisfiability of $U_2$. Since $\widetilde{U}_2$ is sufficiently dense and the constraints are random, it is highly unlikely that there is any assignment satisfying significantly more constraints than a random assignment. See Appendix \ref{appNUGLGSoundnessU2} for the details.\thmspace
\end{proof}

\begin{lemma}\label{lemNUGLGRadiusProperty}
	\lemNUGLGRadiusProperty
\end{lemma}
\begin{proof}[Proof sketch.]
	If $f$ is a partial isomorphism, we can add any vector $z \in Z(v_{i - 1}, v_i)$ to $g^*(v_{i - 1})$ or $g^*(v_{i})$ and $f$ will still preserve the constraints involving $v_{i - 1}$ and $v_i$. Since all edges are good, the set of all such vectors $z$ in each subspace along $p$ spans $\fm$, so we can add vectors at the appropriate places in the path to obtain any desired difference between $g^*(v_0)$ and $g^*(v_n)$. See Appendix \ref{appNUGLGRadiusProperty} for the details.\thmspace
\end{proof}

\begin{lemma}\label{lemNUGLGIndistinguishable}
	$\mathcal{G}(U_1) \equiv_{C^k} \mathcal{G}(U_2)$.
\end{lemma}
\begin{proof}
	It is without loss of generality to assume $H$ is connected, for otherwise Duplicator can apply the strategy presented here on each connected component separately. On every round $i$ of the $k$-pebble bijective game played on $\mathcal{G}(U_1)$ and $\mathcal{G}(U_2)$, for any $u \in V(H)$, let $T_i(u)$ be a minimal tree containing $u$ and all \emph{pebbled vertices} of $H$ (that is, vertices $v \in V(H)$ such that some variable $x_v^g$ is pebbled in one of the two structures) just after Spoiler has picked up a pebble. Let $P_i(u)$ denote the set of all of the vertices in $T_i(u)$ which have degree at least 3 or contain a pebbled vertex, also including $u$. Define $T_i := T_i(u^*_i)$ and $P_i := P_i(u^*_i)$, where $u^*_i$ is the new vertex pebbled in round $i$. Finally, define the forest $F_i(u)$ to be the subgraph of $T_i(u) \setminus T_{i - 1}$ (what this notation means is, remove all edges in $T_{i - 1}$ from $T_i(u)$, then remove isolated vertices) consisting of all segments in $T_i(u)$ between vertices in $P_i(u) \cup V(T_{i - 1})$ which have length less than $r$. See Figure \ref{figBigTree} for an example.
	
	\ipns{.75}{BigTree}{The tree $T_i(u)$.}{\label{figBigTree}The tree consisting of all vertices and edges in the figure is $T_i(u)$. This is a minimal tree that includes all pebbled vertices, which are filled in red, and vertex $u$, which is near the top left corner. The green dashed line outlines the boundary of $T_{i - 1}$ (not all vertices and edges of this tree are shown, just those that intersect $T_i(u)$). Assuming that $r = 3$ (which is not nearly large enough for this many pebbles; this is just for the purpose of illustration), the forest $F_i(u)$ is as depicted in blue, consisting of the lettered vertices $A$ through $H$ and all of the edges between those vertices.}
	
	We need the following lemma, which is proved in Appendix \ref{appGirth}. The proof uses the fact that $H$ has girth at least $(k + 1)^2 r$.\thmspace
	\begin{lemma}\label{lemGirth}
		\lemGirth
	\end{lemma}

	Let $X_i(u)$ denote the variable set of $\mathcal{G}(U_1)$ and $\mathcal{G}(U_2)$ restricted to $T_i(u)$,
	$$X_i(u) := \{x_v^g \suchthat v \in T_i(u),\ g \in \fm\}.$$
	On each round $i$, Duplicator's strategy is to define functions
	$$g^*(i, u, \cdot): V(T_i(u)) \to \fm$$
	for each $u \in V(H)$, satisfying the following two properties:
	\begin{enumerate}[label={(\arabic*)}]
		\item\label{itmNUGLGProp1} For any pebbled vertex $v \in V(H)$, $g^*(i, u, v) = g^*(i - 1, u^*_{i - 1}, v)$.
		\item\label{itmNUGLGProp2} The map $f_{i, u}: X_i(u) \to X_i(u)$ defined by
		$$f_{i, u}(x_v^g) := x_v^{g + g^*(i, u, v)}$$
		gives a partial isomorphism between $\mathcal{G}(U_1)$ and $\mathcal{G}(U_2)$.
	\end{enumerate}
	Duplicator then presents Spoiler with the bijection
	$$f_i(x_v^g) := x_v^{g + g^*(i, v, v)},$$
	which is valid by property \ref{itmNUGLGProp1}. No matter which vertex $u^*_i$ Spoiler chooses, the map $f_{i, u^*_i}$ agrees with $f_i$ over $u^*_i$, so we know that $f_{i, u^*_i}$ respects all pebble pairs since $f_i$ does. Therefore, since the edge between any pair of adjacent pebbled vertices of $H$ must be in $T_i$ (by applying Lemma \ref{lemGirth} to round $i + 1$), Spoiler cannot win, as the map sending each pebbled element in the universe of $\mathcal{G}(U_1)$ to the correspondingly-pebbled element of $\mathcal{G}(U_2)$ is a restriction of $f_{i, u^*_i}$, which is a partial isomorphism by property \ref{itmNUGLGProp2}.
	
	All that remains is to show how Duplicator can satisfy properties \ref{itmNUGLGProp1} and \ref{itmNUGLGProp2} on each round $i$, assuming inductively that they are satisfied on round $i - 1$. Fix a vertex $u \in V(G)$. Duplicator defines $g^*(i, u, \cdot)$ in three steps: first over $V(T_i(u)) \cap V(T_{i - 1})$, then over $V(F_i(u)) \setminus V(T_{i - 1})$, then finally, over the remaining vertices $(V(T_i(u)) \setminus V(T_{i - 1})) \setminus V(F_i(u))$.
	
	Over $V(T_i(u)) \cap V(T_{i - 1})$, Duplicator simply sets
	$$g^*(i, u, v) := g^*(i - 1, u^*_{i - 1}, v),$$
	which is well-defined over $V(T_{i - 1})$ and clearly satisfies both properties \ref{itmNUGLGProp1} and \ref{itmNUGLGProp2}, inductively assuming that $g^*(i - 1, u^*_{i - 1}, \cdot)$ did. Since $V(T_i(u)) \cap V(T_{i - 1})$ contains all pebbled vertices, we no longer have to worry about property \ref{itmNUGLGProp1}; we just have to define $g^*(i, u, \cdot)$ on the remainder of $V(T_i(u))$ so that property \ref{itmNUGLGProp2} is satisfied.
	
	Duplicator then uses the following algorithm to define $g^*(i, u, \cdot)$ over $V(F_i(u)) \setminus V(T_{i - 1})$:
	\begin{algorithm}[H]
		\While{\textbf{true}}
		{
			\uIf{there exists $\{v_1, v_2\} \in E(F_i(u))$ such that $g^*(i, u, v_1)$ is defined but $g^*(i, u, v_2)$ is not defined}
			{
				$g^*(i, u, v_2) \gets g^*(i, u, v_1) + b(v_1, v_2)$\;
			}
			\uElseIf{there exists $v \in V(F_i(u))$ such that $g^*(i, u, v)$ is not defined}
			{
				$g^*(i, u, v) \gets$ anything\;
			}
			\Else
			{
				\KwRet{}\;
			}
		}
	\end{algorithm}
	
	Observe that the constraints involving each edge in $F_i(u)$ considered in the first case are preserved by $f_{i, u}$: for all $g_1, g_2, z \in \fm$,
	\begin{align*} 
		&& x_{v_1}^{g_1} - x_{v_2}^{g_2} = z &\txt{ is an equation in } \mathcal{G}(U_1)\\
		\iff&& (x_{v_1} + g_1) - (x_{v_2} + g_2) = z &\txt{ is an equation in } U_1\\
		\iff&& (x_{v_1} + g_1) - (x_{v_2} + g_2) = z + b(v_1, v_2) &\txt{ is an equation in } U_2\\
		\iff&& (x_{v_1} + g_1 + g^*(i, u, v_1)) \hspace{3cm}\\&& - (x_{v_2} + g_2 + g^*(i, u, v_1) + b(v_1, v_2)) = z &\txt{ is an equation in } U_2\\
		\iff&& x_{v_1}^{g_1 + g^*(i, u, v_1)} - x_{v_2}^{g_2 + g^*(i, u, v_1) + b(v_1, v_2)} = z &\txt{ is an equation in } \mathcal{G}(U_2)\\
		\iff&& x_{v_1}^{g_1 + g^*(i, u, v_1)} - x_{v_2}^{g_2 + g^*(i, u, v_2)} = z &\txt{ is an equation in } \mathcal{G}(U_2)\\
		\iff&& f_{i, u}(x_{v_1}^{g_1}) - f_{i, u}(x_{v_2}^{g_2}) = z &\txt{ is an equation in } \mathcal{G}(U_2).
	\end{align*}
	
	For example, if $F_i(u)$ is as in Figure \ref{figBigTree}, then the first iteration of the algorithm would define $g^*(i, u, B)$ so that the constraints involving $A$ and $B$ are consistent under $f_{i, u}$. The next iteration would then define $g^*(i, u, C)$ so that the constraints involving $B$ and $C$ are consistent. Similarly, the next two iterations would set $g^*(i, u, D)$ and $g^*(i, u, E)$ (these could happen in either order). On the fifth iteration, we would hit the second case of the algorithm and set one of $g^*(i, u, F)$, $g^*(i, u, G)$ or $g^*(i, u, H)$ arbitrarily. The final two iterations would set the other two values according to the first case.
	
	Since the edges encountered in the first case are always made consistent, the only way that $f_{i, u}$ could fail to be a partial isomorphism over $F_i(u)$ is if, at some iteration, there were two different edges satisfying the condition in the first case. Since $F_i(u)$ is a forest, the only way that this could happen is if some connected component of $F_i(u)$ had two distinct vertices $v_1$ and $v_2$ on which $g^*(i, u, \cdot)$ was already defined before the algorithm started, which can only happen if $v_1, v_2 \in V(T_{i - 1})$. But this means that there is a path in $F_i(u)$ from $v_1$ to $v_2$ that violates Lemma \ref{lemGirth}. Thus, property \ref{itmNUGLGProp2} is still satisfied.
	
	At this point, the only remaining edges of $T_i(u)$ which Duplicator needs to worry about are those which are in $T_i(u) \setminus T_{i - 1}$ but are not in $F_i(u)$. By the definition of $F_i(u)$, this consists of paths of length at least $r$, each with a disjoint set of intermediate vertices. Since $g^*(i, u, \cdot)$ has not yet been defined on any of the intermediate vertices, Duplicator can apply Lemma \ref{lemNUGLGRadiusProperty} to each one separately. Thus, \ref{itmNUGLGProp2} is satisfied over the entirety of $X_i(u)$.
\end{proof}

Putting these lemmas together, we can now prove the main result of this chapter.\thmspace

\begin{theorem}\label{thmUGLowGapMain}
	For any constant $\alpha > 0$, there exists a positive integer $q = O(\frac{1}{\alpha^2} \log(\frac{1}{\alpha}))$ such that there is no FPC-definable $\alpha$-approximation algorithm for \textup{\UG($q$)}. This holds even when restricting to \gug instances.
\end{theorem}
\begin{proof}
	Suppose toward a contradiction that there was an FPC-definable $\alpha$-approximation algorithm for \UG($q$), i.e., an FPC-interpretation $\Theta$ of $\tau_\qq$ in $\tauuug$. Let $k := \mu(\Theta)$. Then fix $\gamma := \varepsilon := \frac14$ and use the construction defined in this section to pick a sufficiently high integer $m$. Let $q := 2^m$ (see Appendix \ref{appBigO} for a derivation of the bound on $q$). It follows from Lemmas \ref{lemNUGLGCompleteness}, \ref{lemNUGLGSoundnessU2} and \ref{lemNUGLGIndistinguishable} that, with probability at least $\frac12 - \varepsilon = \frac14$, this construction succeeds in producing a pair of $C^k$-equivalent $\tauuug$-structures, $\mathcal{G}(U_1)$ and $\mathcal{G}(U_2)$ (which are, in fact, \gug instances), whose optimal values differ by a factor of $\alpha$. Specifically, when the construction succeeds, the optimal value of $\mathcal{G}(U_1)$ is $\frac{n}{2^\ell}$ and the optimal value of $\mathcal{G}(U_2)$ is strictly less than $\frac{\alpha n}{2^\ell}$, where $n$ is the total number of constraints. Since the probability of success is nonzero, there is \emph{some} pair of structures produced by this construction satisfying those properties. As $\mathcal{G}(U_1) \equiv_{C^k} \mathcal{G}(U_2)$, $\Theta$ must yield the same value $x \in \qq$ on both instances. Since $\Theta$ gives an $\alpha$-approximation on $\mathcal{G}(U_1)$, we have $\alpha \cdot \frac{n}{2^\ell} \leq x$ (recall the definition from Section \ref{secDefinableInapproximabilityReview}). However, since $\Theta$ gives an $\alpha$-approximation on $\mathcal{G}(U_2)$, we have $x < \frac{\alpha n}{2^\ell}$. We have a contradiction, so no such interpretation $\Theta$ exists.
\end{proof}


\chapter{Conclusion}\label{chaConclusion}

On the surface, the main takeaway from this thesis is, ``Nothing is really different with regard to approximating CSPs when we restrict to FPC-definable algorithms." The best known approximation algorithms turn out to be FPC-definable, and so are the reductions proving that these algorithms are optimal. The status of \ug remains a key missing piece of the puzzle, yet while the existence of a $(1 - \varepsilon, \delta)$ inapproximability gap is unknown, weaker bounds still hold.

However, there are some key respects in which the FPC-definability requirement makes reasoning about approximating CSPs quite different. First, the need for a ``rounding" algorithm to run in polynomial time completely disappears. All that matters is the analysis of such an algorithm---specifically, that it provides an elementary way of computing the optimal value, without breaking the symmetry of the SDP solution matrix. Second, the method for proving lower bounds is completely different, as we no longer rely on the assumption that $\PP \neq \NP$. It is so different that the fundamental problem which is shown to be inexpressible in FPC, distinguishing $\mathcal{G}(U_1)$ from $\mathcal{G}(U_2)$ as defined in Section \ref{secUGLowGap}, is not even $\NP$-hard\footnote{To see this, observe that a given bundle of constraints in either structure is satisfiable if and only if a certain system of $m - \ell$ linear equations over $\fm$ is solvable. In $\mathcal{G}(U_1)$, the union of all of these systems is completely satisfiable, while in $\mathcal{G}(U_2)$, they are not, so distinguishing $\mathcal{G}(U_1)$ from $\mathcal{G}(U_2)$ can be accomplished by Gaussian elimination. }. As such, there is hope that the technique used to prove Theorem \ref{thmUGLowGapMain} can be extended to eventually resolve the FPC-UGC (Conjecture \ref{cnjFPCUGC}) before the ordinary UGC is resolved.

Besides the FPC-UGC, there are several interesting open FPC-approximability questions which are not addressed by this work. Not all problems can be phrased as CSPs in Raghavendra's framework (for example, \vc, \tsp), so the optimal FPC-approximabilities of these problems are yet unknown.



\singlespacing
\bibliographystyle{unsrt} 
\bibliography{bibliography}

\newpage
\appendix
\chapter*{Appendix}
\addcontentsline{toc}{chapter}{Appendix}
\renewcommand{\thesubsection}{\Alph{subsection}}
\numberwithin{theorem}{subsection}

The purpose of this appendix is to give complete proofs of some of the more technical results of Chapter \ref{chaUG} that were omitted due to space constraints.


\subsection{Highly unsatisfiable \gug construction}\label{appUnsatisfiableUGConstruction}

Here we show an explicit (e.g., not randomized) way to construct, for any $\delta > 0$, a \gug instance $U_2$ such that the underlying graph is simple and $U_2$ is not $\delta$-satisfiable, as required in the proof of Theorem \ref{thmKEqualsTwoInapproximability}. Let $n$ be the least integer greater than $\max\{1, \frac{2}{\delta}\}$, and let
$$m := {n \choose 2} = \frac{n(n - 1)}{2}$$
We define $U_2$ over the complete graph on $n$ vertices using the additive group structure on $\fm$, the $m$-dimensional vector space over the finite field with 2 elements. Let $B$ be a basis of $\fm$. For every pair of distinct elements $u, v \in [n]$, assign a distinct basis element $g(u, v) = g(v, u) \in B$. For every pair of distinct variables $x_u$ and $x_v$, $U_2$ has the equation $x_u - x_v = g(u, v)$. Consider a system of equations along any cycle in the underlying graph:
\begin{align*}
	x_{v_1} - x_{v_2} &= g(v_1, v_2)\\
	x_{v_2} - x_{v_3} &= g(v_1, v_2)\\
	\dots\\
	x_{v_{\ell - 1}} - x_{v_\ell} &= g(v_{\ell - 1}, v_\ell)\\
	x_{v_\ell} - x_{v_1} &= g(v_\ell, v_1)
\end{align*}
If all of these equations could be simultaneously satisfied by some assignment, then, adding these equations together, the left-hand sides cancel, so we have
$$0 = g(v_1, v_2) + g(v_1, v_2) + \dots + g(v_{\ell - 1}, v_\ell) + g(v_\ell, v_1).$$
This is impossible, since the fact that cycles have length at least 3 with no repeated vertices implies that each term on the right-hand side is a distinct basis element, so their sum cannot possibly be zero. Thus, no assignment can satisfy any cycle of constraints, so at most a spanning tree of $n - 1$ constraints can be satisfied. Thus, the maximal satisfiability of $U_2$ is
$$(n - 1)/{n \choose 2} = (n - 1)/\left(\frac{n(n - 1)}{2}\right) = \frac{2}{n} < \delta,$$
as desired.


\subsection{Proof of Lemma \ref{lemGSameSatisfiability} (the label-lifted instance has the same satisfiability)}\label{appGSameSatisfiability}

Here we prove Lemma \ref{lemGSameSatisfiability}.\thmspace

\newtheorem*{R0}{Lemma~\ref{lemGSameSatisfiability}}
\begin{R0}
	\lemGSameSatisfiability
\end{R0}
\begin{proof}
	We begin by introducing some notation which is not used outside of this proof. Suppose there are $n$ variables in $U$, denoted $x_{v_1}, x_{v_2}, \dots, x_{v_n}$. For every $i, j \in [n]$, write $c(i, j)$ for the number of constraints between variables $x_{v_i}$ and $x_{v_j}$, and enumerate them as
	$$\{x_{v_i} - x_{v_j} = z_{i, j, k} \suchthat k \in [c(i, j)]\}.$$
	Let $\opt(\cdot)$ denote the optimal value of a \ug instance. In the context of some fixed assignment of variables, for any constraint equation $\beta$ let $\mathbb{I}(\beta)$ be the function that evaluates to 1 if $\beta$ is satisfied under the assignment and 0 if $\beta$ is not satisfied.
	
	Suppose there are $q$ group elements. Since $\mathcal{G}(U)$ contains a factor of $q^2$ more constraints than $U$, to prove that $U$ and $\mathcal{G}(U)$ have the same satisfiability, we must show that
	$$\opt(U) = \frac{1}{q^2}\opt(\mathcal{G}(U)).$$
	
	For one direction, let $x_v$ be an assignment of variables attaining the optimum satisfiability of $U$, i.e.,
	$$\opt(U) = \sum_{i, j \in [n]}\ \sum_{k \in [c(i, j)]} \mathbb{I}(x_{v_i} - x_{v_j} = z_{i, j, k}).$$
	From this, define an assignment of variables of $\mathcal{G}(U)$ by
	$$x_{v}^{g} := x_{v} + g.$$
	Then the optimal value of $\mathcal{G}(U)$ is at least the number of constraints satisfied by this assignment, i.e.,
	\begin{align*} 
		\opt(\mathcal{G}(U)) &\geq \sum_{i, j \in [n]}\ \sum_{g_i, g_j \in A}\ \sum_{k \in [c(i, j)]} \mathbb{I}((x_{v_i}^{g_i} - g_i) - (x_{v_j}^{g_j} - g_j) = z_{i, j, k})\\
		&= \sum_{i, j \in [n]}\ \sum_{g_i, g_j \in A}\ \sum_{k \in [c(i, j)]} \mathbb{I}(((x_{v_i} + g_i) - g_i) \push\hspace{3cm} - ((x_{v_j} + g_j) - g_j) = z_{i, j, k})\\
		&= \sum_{i, j \in [n]}\ \sum_{g_i, g_j \in A}\ \sum_{k \in [c(i, j)]} \mathbb{I}(x_{v_i} - x_{v_j} = z_{i, j, k})\\
		&= \sum_{i, j \in [n]} (q^2) \sum_{k \in [c(i, j)]} \mathbb{I}(x_{v_i} - x_{v_j} = z_{i, j, k})\\
		&= q^2 \opt(U).
	\end{align*}
	Rearranging, we have
	$$\opt(U) \leq \frac{1}{q^2}\opt(\mathcal{G}(U)).$$
	
	For the other direction, let $x_v^g$ be an assignment of variables attaining the optimum satisfiability of $\mathcal{G}(U)$. Then
	\begin{align*} 
		\opt(\mathcal{G}(U)) &= \sum_{i, j \in [n]}\ \sum_{g_i, g_j \in A}\ \sum_{k \in [c(i, j)]} \mathbb{I}((x_{v_i}^{g_i} - g_i) - (x_{v_j}^{g_j} - g_j) = z_{i, j, k})\\
		&= \frac{1}{q^{n - 2}} \sum_{i, j \in [n]}\ \sum_{\seq{g}{n} \in A}\ \sum_{k \in [c(i, j)]} \mathbb{I}((x_{v_i}^{g_i} - g_i) - (x_{v_j}^{g_j} - g_j) = z_{i, j, k}),
	\end{align*}
	since, for every fixed $i, j \in [n]$, the term
	$$\sum_{k \in [c(i, j)]} \mathbb{I}((x_{v_i}^{g_i} - g_i) - (x_{v_j}^{g_j} - g_j) = z_{i, j, k})$$
	is counted exactly $q^{n - 2}$ times. Rearranging the order of summation, we have
	\begin{equation}\label{equGSameSatisfiability1}
		\opt(\mathcal{G}(U)) = \frac{1}{q^{n - 2}} \sum_{\seq{g}{n} \in A}\ \sum_{i, j \in [n]}\ \sum_{k \in [c(i, j)]} \mathbb{I}((x_{v_i}^{g_i} - g_i) - (x_{v_j}^{g_j} - g_j) = z_{i, j, k})
	\end{equation}
	Therefore, there must be some fixed $\seq{g}{n} \in A$ such that
	\begin{equation}\label{equGSameSatisfiability2}
		\sum_{i, j \in [n]}\ \sum_{k \in [c(i, j)]} \mathbb{I}((x_{v_i}^{g_i} - g_i) - (x_{v_j}^{g_j} - g_j) = z_{i, j, k}) \geq \frac{1}{q^{2}} \opt(\mathcal{G}(U)),
	\end{equation}
	for otherwise, if all of the $q^n$ choices of $\seq{g}{n} \in A$ failed to satisfy (\ref{equGSameSatisfiability2}), we could strictly upper-bound the right-hand side of (\ref{equGSameSatisfiability1}) by
	$$\frac{1}{q^{n - 2}}(q^n)\left(\frac{1}{q^{2}} \opt(\mathcal{G}(U))\right) = \opt(\mathcal{G}(U)),$$
	contradicting (\ref{equGSameSatisfiability1}). Using these fixed $g_i$ values, we define an assignment of variables of $U$ by
	$$x_{v_i} := x_{v_i}^{g_i} - g_i.$$
	It then follows that the optimal value of $U$ is at least the number of constraints satisfied by this assignment, i.e.,
	\begin{align*}
		\opt(U) &\geq \sum_{i, j \in [n]}\ \sum_{k \in [c(i, j)]} \mathbb{I}(x_{v_i} - x_{v_j} = z_{i, j, k})\\
		&= \sum_{i, j \in [n]}\ \sum_{k \in [c(i, j)]} \mathbb{I}((x_{v_i}^{g_i} - g_i) - (x_{v_j}^{g_j} - g_j) = z_{i, j, k})\\
		&\geq \frac{1}{q^{2}} \opt(\mathcal{G}(U)) \stext{by (\ref{equGSameSatisfiability2})},
	\end{align*}
	as desired.
\end{proof}

We remark that, with very slight modification, this argument also shows that the $G$ operator of Atserias and Dawar \cite{DefinableInapproximabilityJournal} preserves the exact satisfiability of a \threexor instance. In other words, part (2) of Lemma 3 of \cite{DefinableInapproximabilityJournal} can be strengthened, and as a consequence, the third paragraph in the proof of Lemma 4 of \cite{DefinableInapproximabilityJournal} is unnecessary.


\subsection{Cops and robbers construction}\label{appCopsAndRobbers}

Here we show how to construct a graph $H = H_k$ satisfying the following four properties stated in Section \ref{secUGExact}:
\begin{enumerate}[label={(\arabic*)}]
	\item $H$ is connected.
	\item $H$ is 3-regular.
	\item $H$ is bipartite.
	\item The robber player wins the $(k - 1)$-cop edge-robber game.
\end{enumerate}

Start with the complete graph on $k$ vertices. Replace every vertex with a cycle of $2(k - 1)$ vertices, and replace every edge with two ``bridge" edges joining distinct pairs of adjacent vertices in each cycle, as in Figure \ref{figCopsAndRobbersConstruction}.\thmspace

\counterwithin{figure}{subsection}
\ipns{.9}{CopsAndRobbersConstruction}{An example of the construction of $H$ where $k = 4$.}{\label{figCopsAndRobbersConstruction}An example of the construction of $H$ where $k = 4$. The bridge edges are drawn in red.}

It is easily verified that $H$ is connected, 3-regular, and bipartite. The robber player's strategy is to always have the robber occupy one of the $k$ cycles without one of the $k - 1$ cops in it. After a cop moves into the cycle occupied by the robber, the robber moves around its cycle to one of the bridges to an unoccupied cycle and crosses over. Since there is only one cop in the robber's cycle, it cannot block both bridges. Thus, the robber player can infinitely avoid capture.


\subsection{Proof of Lemma \ref{lemUGExactSoundness} (soundness of $U_2$)}\label{appParityArgument}

Here we prove Lemma \ref{lemUGExactSoundness}.\thmspace

\newtheorem*{R1}{Lemma~\ref{lemUGExactSoundness}}
\begin{R1}
	\lemUGExactSoundness
\end{R1}
\begin{proof}
	Suppose we have some assignment of variables and we change the value of one of these variables, $x_v$, by adding some group element $g \in A$ to it. Let $u_1, u_2, u_3 \in V(H)$ be the three neighbours of $v$ in $H$. If $v \notin \{v_1^*, v_2^*\}$, then the list of constraints involving $x_v$ (up to a relabeling of $u_1$, $u_2$ and $u_3$) is
	\begin{align*}
		x_{v} + x_{u_1} &= e\\
		x_{v} + x_{u_1} &= a\\
		x_{v} + x_{u_2} &= e\\
		x_{v} + x_{u_2} &= b\\
		x_{v} + x_{u_3} &= e\\
		x_{v} + x_{u_3} &= c,
	\end{align*}
	If $g = e$ then nothing changes. Otherwise, without loss of generality, assume $g = a$. After adding $a$ to $x_v$, if the value of $x_{v} + x_{u_2}$ was in the set $\{e, a\}$, it will still be in $\{e, a\}$, and if the value was \emph{not} in $\{e, a\}$, then it will still \emph{not} be. Thus, after adding $a$ to $x_v$, the satisfiability of the first pair of equations will remain the same. For the second pair of equations, if the value of $x_{v} + x_{u_2}$ was in the set $\{e, b\}$, then after adding $a$ the value will be in the set $\{a, c\}$, and vice versa. Thus, either one of these two equations will become satisfied or one of these two equations will become unsatisfied; in other words, the satisfiability will change by one. Similarly, for the second pair of equations, if the value of $x_{v} + x_{u_3}$ was in the set $\{e, c\}$, then after adding $a$ the value will be in the set $\{a, b\}$, and vice versa, so again, the satisfiability will change by one. Thus, the total satisfiability of the 6 equations involving the variable $x_v$ will either remain the same or change by 2.
	
	On the other hand, in the special case where $v \in \{v_1^*, v_2^*\}$, the list of constraints involving $x_v$ (up to a relabeling of $u_1$, $u_2$ and $u_3$) is
	\begin{align*}
		x_{v} + x_{u_1} &= b\\
		x_{v} + x_{u_1} &= c\\
		x_{v} + x_{u_2} &= e\\
		x_{v} + x_{u_2} &= b\\
		x_{v} + x_{u_3} &= e\\
		x_{v} + x_{u_3} &= c.
	\end{align*}
	If $g = e$ then nothing changes. If $g = a$ then the satisfiability of the first pair of equations will remain the same (since adding $a$ takes $b$ to $c$ and $c$ to $b$), and in the second and third pairs of equations the satisfiabilities will each change by one, for a total change of 0 or 2, as in the previous case. If $g = b$ or $g = c$, one can analogously check that the satisfiability of the 6 equations again changes by 0 or 2.
	
	Thus, in all cases, changing $x_v$ by adding any group element $g$ preserves the parity of the number of constraints of $U_2$ that are satisfied. There are a total of $2\abs{E(H)}$ constraints between variables of $U_2$, coming in inconsistent pairs. Since the assignment $x_v := e$ satisfies exactly $\abs{E(H)} - 1$ of these constraints (one from every pair, except none from the pair between $x_{v_1^*}$ and $x_{v_2^*}$), and every time a variable assignment is changed, the parity stays the same, no assignment can satisfy exactly $\abs{E(H)}$ constraints. Thus, no assignment can satisfy \emph{at least} $\abs{E(H)}$ constraints, since that is the maximum number that can possibly be satisfied.
\end{proof}


\subsection{Proof of Lemma \ref{lemUGExactClaim} (Duplicator's invariant)}\label{appUGExactClaim}

Here we prove Lemma \ref{lemUGExactClaim}. We begin by recording some easy but important observations about these two instances.\thmspace

\begin{lemma}\label{lemUGExactSpecialProperty}
	The variable set of both $\mathcal{G}(U_1)$ and $\mathcal{G}(U_2)$ is
	$$\{x_v^g \suchthat v \in V(H),\ g \in A\}.$$
	In both instances, for every $v_1, v_2 \in V(H)$, there exist distinct $y_1, y_2 \in A$ such that, for all $g_1, g_2 \in A$, the constraints between the pair of variables $x_{v_1}^{g_1}$ and $x_{v_2}^{g_2}$ are of the form
	\begin{align*}
		x_{v_1}^{g_1} + x_{v_2}^{g_2} = g_1 + g_2 + y_1 && \txt{and} && x_{v_1}^{g_1} + x_{v_2}^{g_2} = g_1 + g_2 + y_2,
	\end{align*}
	where $y_1 + y_2 = m(\{v_1, v_2\})$.
\end{lemma}
\begin{proof}
	This follows from inspection of the definitions of $U_1$ and $U_2$ from Section \ref{secUGExact}.
\end{proof}\thmspace

Now we can prove Lemma \ref{lemUGExactClaim}.\thmspace

\newtheorem*{R2}{Lemma~\ref{lemUGExactClaim}}
\begin{R2}
	\lemUGExactClaim
\end{R2}
\begin{proof}
	We proceed by induction on $i$. For the base case ($i = 0$), recall that $f_0$ is the identity map and $r_0 = \{v^*_1, v^*_2\}$. Since $\mathcal{G}(U_1)$ and $\mathcal{G}(U_2)$ agree everywhere except on the relations between variables involving the vertices of $r_0$, condition \ref{itmUGExactGood2} holds. Between a pair of vertices $x_{v^*_1}^{g_1}$ and $x_{v^*_2}^{g_2}$ for $g_1, g_2 \in A$, the constraints in $\mathcal{G}(U_1)$ are
	$$x_{v^*_1}^{g_1} + x_{v^*_2}^{g_2} = g_1 + g_2 \txt{ and } x_{v^*_1}^{g_1} + x_{v^*_2}^{g_2} = g_1 + g_2 + a,$$
	while the corresponding constraints in $\mathcal{G}(U_2)$ are
	$$x_{v^*_1}^{g_1} + x_{v^*_2}^{g_2} = g_1 + g_2 + b \txt{ and } x_{v^*_1}^{g_1} + x_{v^*_2}^{g_2} = g_1 + g_2 + c.$$
	Since $g_1 + g_2$, $g_1 + g_2 + a$, $g_1 + g_2 + b$ and $g_1 + g_2 + c$ are all distinct, condition \ref{itmUGExactGood3} holds as well.
	
	Now fix some $i \geq 1$ and suppose that \ref{itmUGExactGood2} and \ref{itmUGExactGood3} hold for $i - 1$. If the robber does not move, then $f_i = f_{i - 1}$, so there is nothing to prove. So suppose that the robber does move, i.e., $r_i \neq r_{i - 1}$. To prove that $f_i$ satisfies \ref{itmUGExactGood2} for an arbitrary pair of vertices $\{v_1, v_2\} \neq r_i$, there are three cases to consider. It may be helpful for the reader to refer back to Figure \ref{figUGExactDiagram2} from Section \ref{secUGExact}.
	
	Case 1: $\{v_1, v_2\} = \{p_0, p_1\} = r_{i - 1}$. Say that $v_1 = p_0$ and $v_2 = p_1$. Let $y_1, y_2, y_2, y_4 \in A$ be as in Lemma \ref{lemUGExactSpecialProperty}, so that for each $g_1, g_2 \in A$, the two constraints in $\mathcal{G}(U_1)$ between $x_{v_1}^{g_1}$ and $x_{v_2}^{g_2}$ are
	\begin{align*}
		x_{v_1}^{g_1} + x_{v_2}^{g_2} = g_1 + g_2 + y_1 && \txt{and} && x_{v_1}^{g_1} + x_{v_2}^{g_2} = g_1 + g_2 + y_2,
	\end{align*}
	and the constraints in $\mathcal{G}(U_2)$ between $f_{i - 1}(x_{v_1}^{g_1})$ and $f_{i - 1}(x_{v_2}^{g_2})$ are
	\begin{align*}
		&&f_{i - 1}(x_{v_1}^{g_1}) + f_{i - 1}(x_{v_2}^{g_2}) &= g_1 + g_2 + y_3\\ \txt{and} &&f_{i - 1}(x_{v_1}^{g_1}) + f_{i - 1}(x_{v_2}^{g_2}) &= g_1 + g_2 + y_4.
	\end{align*}
	Since \ref{itmUGExactGood3} held for $f_{i - 1}$, it follows that the right-hand sides of all four equations are all distinct, so $A = \{y_1, y_2, y_3, y_4\}$. From the group addition law in the Klein four-group, it follows that $y_1 + y_2 + y_3 = y_4$. Since $f_{i}(x_{v_1}^{g_1}) = f_{i - 1}(x_{v_1}^{g_1})$ and $f_{i}(x_{v_2}^{g_2}) = f_{i - 1}(x_{v_2}^{g_2 + m(e_1)})$, the two constraints in $\mathcal{G}(U_2)$ between $f_{i}(x_{v_1}^{g_1})$ and $f_{i}(x_{v_2}^{g_2})$ are
	\begin{align*}
		&&f_{i}(x_{v_1}^{g_1}) + f_{i}(x_{v_2}^{g_2}) &= g_1 + g_2 + m(e_1) + y_3 \\ \txt{and} && f_{i}(x_{v_1}^{g_1}) + f_{i}(x_{v_2}^{g_2}) &= g_1 + g_2 + m(e_1) + y_4.
	\end{align*}
	To show that these constraints are the same as those in $\mathcal{G}(U_1)$, we must argue that $m(e_1) + y_3$ and $m(e_1) + y_4$ are both in the set $\{y_1, y_2\}$. We give the proof for $m(e_1) + y_3$; the proof for $m(e_1) + y_4$ is completely analogous. Suppose first that $m(e_1) + y_3 = y_3$. This is a contradiction because $m$ never takes on the value of the identity $e \in A$. Suppose instead that $m(e_1) + y_3 = y_4$. This implies that
	$$m(e_1) = y_1 + y_2 = m(\{v_1, v_2\}),$$
	where the second equality follows from Lemma \ref{lemUGExactSpecialProperty}. This is a contradiction since $e_1$ and $\{v_1, v_2\}$ are different edges incident to the same vertex $p_1$, so they must have different values under $m$. Thus, the only remaining possibilities are that $m(e_1) + y_3 = y_1$ or $m(e_1) + y_3 = y_2$, as desired.
	
	Case 2: $\{v_1, v_2\} = \{p_i, p_{i + 1}\}$ for $1 \leq i < \ell$. Say that $v_1 = p_i$ and $v_2 = p_{i + 1}$. Let $y_1, y_2, y_2, y_4 \in A$ be as in Lemma \ref{lemUGExactSpecialProperty}, so that for each $g_1, g_2 \in A$, the two constraints in $\mathcal{G}(U_1)$ between $x_{v_1}^{g_1}$ and $x_{v_2}^{g_2}$ are
	\begin{align*}
		x_{v_1}^{g_1} + x_{v_2}^{g_2} = g_1 + g_2 + y_1 && \txt{and} && x_{v_1}^{g_1} + x_{v_2}^{g_2} = g_1 + g_2 + y_2,
	\end{align*}
	and the constraints in $\mathcal{G}(U_2)$ between $f_{i - 1}(x_{v_1}^{g_1})$ and $f_{i - 1}(x_{v_2}^{g_2})$ are
	\begin{align*}
		&&f_{i - 1}(x_{v_1}^{g_1}) + f_{i - 1}(x_{v_2}^{g_2}) &= g_1 + g_2 + y_3 \\ \txt{and} && f_{i - 1}(x_{v_1}^{g_1}) + f_{i - 1}(x_{v_2}^{g_2}) &= g_1 + g_2 + y_4.
	\end{align*}
	Since \ref{itmUGExactGood2} held for $f_{i - 1}$, it follows that $\{y_1, y_2\} = \{y_3, y_4\}$. As $f_{i}(x_{v_1}^{g_1}) = f_{i - 1}(x_{v_1}^{g_1 + m(e_i)})$ and $f_{i}(x_{v_2}^{g_2}) = f_{i - 1}(x_{v_2}^{g_2 + m(e_{i + 1})})$, the two constraints in $\mathcal{G}(U_2)$ between $f_{i}(x_{v_1}^{g_1})$ and $f_{i}(x_{v_2}^{g_2})$ are
	\begin{align*}
		&&f_{i}(x_{v_1}^{g_1}) + f_{i}(x_{v_2}^{g_2}) &= g_1 + g_2 + m(e_i) + m(e_{i + 1}) + y_3 \\ \txt{and} &&f_{i}(x_{v_1}^{g_1}) + f_{i}(x_{v_2}^{g_2}) &= g_1 + g_2 + m(e_i) + m(e_{i + 1}) + y_4.
	\end{align*}
	As in the previous case, we must show that $m(e_i) + m(e_{i + 1}) + y_3$ and $m(e_i) + m(e_{i + 1}) + y_4$ are both in the set $\{y_1, y_2\}$. If $m(e_i) = m(e_{i + 1})$, then they cancel, and the result then follows from the fact that $\{y_3, y_4\} = \{y_1, y_2\}$. Otherwise, they are distinct nontrivial elements of $A$, and since they both share common vertices with the edge $\{v_1, v_2\}$, they are also distinct from the nontrivial element $m(\{v_1, v_2\})$. This means that
	$$m(e_i) + m(e_{i + 1}) = m(\{v_1, v_2\}) = y_3 + y_4,$$
	where the second equality follows from Lemma \ref{lemUGExactSpecialProperty}. Therefore,
	$$m(e_i) + m(e_{i + 1}) + y_3 = y_3 + y_3 + y_4 = y_4 \in \{y_3, y_4\} = \{y_1, y_2\},$$
	and analogously,
	$$m(e_i) + m(e_{i + 1}) + y_4 = y_3 + y_4 + y_4 = y_3 \in \{y_3, y_4\} = \{y_1, y_2\},$$
	as desired.
	
	Case 3: $\{v_1, v_2\}$ is not on the path from $r_{i - 1}$ to $r_i$ (this is the case discussed in the example from Section \ref{secUGExact}). The only other edges of $H$ we have to worry about are those which are incident to a vertex in $H$ over which $f_{i - 1}$ and $f_{i}$ differ. These are precisely the edges $e_i$, for $i \in [r]$, so assume that $\{v_1, v_2\} = e_i$ where $v_1 = p_i$. Again, let $y_1, y_2, y_2, y_4 \in A$ be as in Lemma \ref{lemUGExactSpecialProperty}, so that for each $g_1, g_2 \in A$, the two constraints in $\mathcal{G}(U_1)$ between $x_{v_1}^{g_1}$ and $x_{v_2}^{g_2}$ are
	\begin{align*}
		x_{v_1}^{g_1} + x_{v_2}^{g_2} = g_1 + g_2 + y_1 && \txt{and} && x_{v_1}^{g_1} + x_{v_2}^{g_2} = g_1 + g_2 + y_2,
	\end{align*}
	and the constraints in $\mathcal{G}(U_2)$ between $f_{i - 1}(x_{v_1}^{g_1})$ and $f_{i - 1}(x_{v_2}^{g_2})$ are
	\begin{align*}
		&&f_{i - 1}(x_{v_1}^{g_1}) + f_{i - 1}(x_{v_2}^{g_2}) &= g_1 + g_2 + y_3 \\ \txt{and} && f_{i - 1}(x_{v_1}^{g_1}) + f_{i - 1}(x_{v_2}^{g_2}) &= g_1 + g_2 + y_4.
	\end{align*}
	Since \ref{itmUGExactGood2} held for $f_{i - 1}$, it follows that $\{y_1, y_2\} = \{y_3, y_4\}$. Since $f_{i}(x_{v_1}^{g_1}) = f_{i - 1}(x_{v_1}^{g_1 + m(e_i)})$ and $f_{i}(x_{v_2}^{g_2}) = f_{i - 1}(x_{v_2}^{g_2})$, the two constraints in $\mathcal{G}(U_2)$ between $f_{i}(x_{v_1}^{g_1})$ and $f_{i}(x_{v_2}^{g_2})$ are
	\begin{align*}
		&&f_{i}(x_{v_1}^{g_1}) + f_{i}(x_{v_2}^{g_2}) &= g_1 + g_2 + m(e_i) + y_3 \\ \txt{and} &&f_{i}(x_{v_1}^{g_1}) + f_{i}(x_{v_2}^{g_2}) &= g_1 + g_2 + m(e_i) + y_4.
	\end{align*}
	As in the previous cases, we must show that $m(e_i) + y_3$ and $m(e_i) + y_4$ are both in the set $\{y_1, y_2\}$. This follows from the fact that $m(e_i) = y_3 + y_4$ by Lemma \ref{lemUGExactSpecialProperty}, so
	$$\{m(e_i) + y_3, m(e_i) + y_4\} = \{y_3 + y_4 + y_3, y_3 + y_4 + y_4\} = \{y_4, y_3\} = \{y_1, y_2\}$$
	as desired.
	
	That concludes the proof of \ref{itmUGExactGood2}. To prove \ref{itmUGExactGood3}, let $\{v_1, v_2\} = r_i$, where $v_1 = p_\ell$ and $v_2 = p_{\ell + 1}$, and again let $y_1, y_2, y_2, y_4 \in A$ be as in Lemma \ref{lemUGExactSpecialProperty}, so that for each $g_1, g_2 \in G$, the two constraints in $\mathcal{G}(U_1)$ between $x_{v_1}^{g_1}$ and $x_{v_2}^{g_2}$ are
	\begin{align*}
		x_{v_1}^{g_1} + x_{v_2}^{g_2} = g_1 + g_2 + y_1 && \txt{and} && x_{v_1}^{g_1} + x_{v_2}^{g_2} = g_1 + g_2 + y_2,
	\end{align*}
	and the constraints in $\mathcal{G}(U_2)$ between $f_{i - 1}(x_{v_1}^{g_1})$ and $f_{i - 1}(x_{v_2}^{g_2})$ are
	\begin{align*}
		&&f_{i - 1}(x_{v_1}^{g_1}) + f_{i - 1}(x_{v_2}^{g_2}) &= g_1 + g_2 + y_3 \\ \txt{and} && f_{i - 1}(x_{v_1}^{g_1}) + f_{i - 1}(x_{v_2}^{g_2}) &= g_1 + g_2 + y_4.
	\end{align*}
	Since \ref{itmUGExactGood2} held for $f_{i - 1}$, it follows that $\{y_1, y_2\} = \{y_3, y_4\}$. Since $f_{i}(x_{v_1}^{g_1}) = f_{i - 1}(x_{v_1}^{g_1 + m(e_\ell)})$ and $f_{i}(x_{v_2}^{g_2}) = f_{i - 1}(x_{v_2}^{g_2})$, the two constraints in $\mathcal{G}(U_2)$ between $f_{i}(x_{v_1}^{g_1})$ and $f_{i}(x_{v_2}^{g_2})$ are
	\begin{align*}
		&&f_{i}(x_{v_1}^{g_1}) + f_{i}(x_{v_2}^{g_2}) &= g_1 + g_2 + m(e_\ell) + y_3 \\ \txt{and} &&f_{i}(x_{v_1}^{g_1}) + f_{i}(x_{v_2}^{g_2}) &= g_1 + g_2 + m(e_\ell) + y_4.
	\end{align*}
	Now we must show that $m(e_\ell) + y_3$ and $m(e_\ell) + y_4$ are both \emph{not} in the set $\{y_1, y_2\}$. The proof is analogous to Case 1 from above, and we only show the first part, that $m(e_\ell) + y_3$ is not in $\{y_1, y_2\} = \{y_3, y_4\}$. Suppose first that $m(e_\ell) + y_3 = y_3$. This is a contradiction because $m$ never takes on the value of the identity $e \in A$. Suppose instead that $m(e_\ell) + y_3 = y_4$. This implies that
	$$m(e_\ell) = y_3 + y_4 = m(\{v_1, v_2\}),$$
	where the second equality follows from Lemma \ref{lemUGExactSpecialProperty}. This is a contradiction since $e_\ell$ and $\{v_1, v_2\}$ are different edges incident to the same vertex $p_\ell$, so they must have different values under $m$. Thus, $m(e_\ell) + y_3 \notin \{y_3, y_4\} = \{y_1, y_2\}$; the proof that $m(e_\ell) + y_4 \notin \{y_1, y_2\}$ is similar.
\end{proof}


\subsection{Proof of Lemma \ref{lemNUGLGMostEdgesGood} (most edges are good edges)}\label{appNUGLGMostEdgesGood}

Here we prove Lemma \ref{lemNUGLGMostEdgesGood}. First, we need the following two sub-lemmas. The proof of the first one is inspired by \cite{Stackexchange}.\thmspace

\begin{lemma}\label{lemProbRandomVectorsSpan}
	For any two positive integers $m$ and $n$, the probability that $n$ vectors in $\fm$, chosen independently and uniformly at random, fail to span $\fm$ is at most $2^{m - n}$.
\end{lemma}
\begin{proof}
	It is well known \cite[Sec.\txt{} III.4]{LinearAlgebraBook} that every $d$-dimensional subspace of a vector space of dimension $m$ has a unique complement subspace of dimension $m - d$. Therefore, since there are exactly $2^m - 1$ one-dimensional subspaces of $\fm$, there are exactly $2^m - 1$ subspaces of dimension $m - 1$.
	
	The probability that a randomly chosen vector lies within a given $(m - 1)$-dimensional subspace is $\frac{1}{2}$. As the vectors are chosen independently, the probability that all $n$ vectors lie within any given $(m - 1)$-dimensional subspace is $2^{-n}$. Since there are at most $2^m$ different $(m - 1)$-dimensional subspaces of $\fm$, by the union bound, the probability that all $n$ vectors lie within \emph{some} $(m - 1)$-dimensional subspace is at most $2^m \cdot 2^{-n} = 2^{m - n}$. Therefore, the probability that the $n$ vectors fail to span $\fm$ is bounded by $2^{m - n}$, since the only way this can happen is if they lie within some $(m - 1)$-dimensional subspace of $\fm$.\thmspace
\end{proof}

\begin{lemma}\label{lemCountPaths}
	For any positive integers $d$ and $r$, in a $d$-regular graph of girth greater than $r$ there are exactly $r(d - 1)^{r - 1}$ distinct paths of length $r$ passing through any given edge.
\end{lemma}
\begin{proof}
	Fix an edge $e_0 = \{u_1, u_2\}$. To enumerate all of the ways in which we can choose a path $p = v_0, \seq{v}{r}$ passing through $e_0$, we first orient $p$ so that $u_2$ occurs at a greater index in $p$ than $u_1$, ensuring that we do not double-count a path and its reverse. The first choice we make is the position in the path where $e_0$ lies, i.e., the index of $u_1$ in $p$. There are $r$ such choices of index, since it is impossible to have $u_1 = v_r$. We choose each of the remaining $r - 1$ edges by growing the path out from $e_0$. Since the girth is greater than $r$, there are no constraints about repeating vertices to worry about, so at each step, there are exactly $d - 1$ neighbors to choose from. Thus, the total number of paths is $r(d - 1)^{r - 1}$.
\end{proof}

We can now prove Lemma \ref{lemNUGLGMostEdgesGood}.\thmspace
\newtheorem*{R3}{Lemma~\ref{lemNUGLGMostEdgesGood}}
\begin{R3}
	\lemNUGLGMostEdgesGood
\end{R3}
\begin{proof}
	By Lemma \ref{lemProbRandomVectorsSpan}, the probability that the vectors in the $Z$-subspaces along a given path of length $r$ fail to span $\fm$ is at most $2^{m - r\ell}$. Since, by Lemma \ref{lemCountPaths}, there are $r(d - 1)^{r - 1} \leq rd^r$ paths of length $r$ through any given edge $e_0$, it follows from the union bound that the probability that $e_0$ is a bad edge is at most
	\begin{align*}
		rd^r \cdot 2^{m - r\ell} &= \exp(\ln(rd^r \cdot 2^{m - r\ell}))\stext{where $\exp(x) \equiv e^x$}\\
		&= \exp(\ln(r) + r \ln(d) + (m - r\ell)\ln(2))\\
		&\leq \exp(r + r \ln(d) + (m - r\ell)\ln(2))\\
		&= \exp(r((\ln(d) + 1) - \ell\ln(2)) + m\ln(2))\\
		&\leq \exp\left(r((\ln(d) + 1) - \left(\log_2(d) + 2\log_2(e)\right)\ln(2)) + m\ln(2)\right) \stextn{from (\ref{equNUGLGParameterLDef})}\\
		&= \exp\left(r((\ln(d) + 1) - \left(\frac{\ln(d) + 2}{\ln(2)}\right)\ln(2)) + m\ln(2)\right)\\
		&= \exp(m\ln(2) - r)\\
		&\leq \exp(m\ln(2) - (m\ln(2) - \ln(\gamma))) \stext{from (\ref{equNUGLGParameterRGeq})}\\
		&= \exp(\ln(\gamma))\\
		&= \gamma.
	\end{align*}
	Therefore, the expected fraction of bad edges of $\widetilde{H}$ is at most $\gamma$. With probability at least $\frac12$, the fraction of bad edges in $\widetilde{H}$ is less than or equal to this expectation.
\end{proof}


\subsection{Proof of Lemma \ref{lemNUGLGSoundnessU2} (soundness of $\widetilde{U}_2$ and thus $U_2$)}\label{appNUGLGSoundnessU2}

Here we prove Lemma \ref{lemNUGLGSoundnessU2}. We begin by showing that, with high probability, $\widetilde{U}_2$ is highly unsatisfiable. The central proof technique used here, applying Hoeffding's inequality and the union bound, is used by Atserias and Dawar \cite[Lemma 4]{DefinableInapproximabilityJournal} to argue that a random \threexor instance is probably only slightly more than $\frac12$-satisfiable. The main challenge in adapting this technique is that our domain has size $2^m$ instead of $2$, so there are far too many assignments to consider. To circumvent this obstacle, we only consider those assignments which satisfy a spanning tree of constraints.\thmspace

\begin{lemma}\label{lemNUGLGSoundnessUTilde}
	With probability at least $1 - \varepsilon$, the satisfiability of $\widetilde{U}_2$ is less than $\left(1 - \gamma\right) \left(\frac{\alpha}{2^{\ell}}\right)$.
\end{lemma}
\begin{proof}
	Say that a bundle of constraints in $\widetilde{U}_2$ is \emph{satisfied} by a given assignment of variables if one of the $2^\ell$ constraints in the bundle is satisfied. Consider the following nondeterministic algorithm for satisfying a maximal number of constraints in $\widetilde{U}_2$, where $v_0$ is an arbitrarily chosen vertex of $\widetilde{H}$:
	\begin{enumerate}[label={(\arabic*)}]
		\item\label{itmNUGLGSoundnessU3PickTree} Nondeterministically choose a spanning tree $T \subseteq E(\widetilde{H})$.
		\item\label{itmNUGLGSoundnessU3PickVectorsFromBundles} For each $\{v_1, v_2\} \in T$, nondeterministically choose a vector $z^*(v_1, v_2) \in Z(v_1, v_2)$.
		\item\label{itmNUGLGSoundnessU3Assign} Assign $x_{v_0} := 0$, then assign all of the other variables so that, for all $\{v_1, v_2\} \in T$, $x_{v_1} + x_{v_2} = z^*(v_1, v_2)$ (this assignment is unique after fixing $x_{v_0}$, and can be defined inductively through the edges of $T$).
	\end{enumerate}
	Note that it is without loss of generality to assume $x_{v_0} = 0$ under any optimal assignment of variables, for if it was not, we could subtract $x_{v_0}$ from all of the variables and the exact same set of constraints would be satisfied. By similar reasoning, the set of edges of $\widetilde{H}$ whose bundles are satisfied under a given optimal assignment must contain a spanning tree of $\widetilde{H}$, for if it contained two connected components separated by an edge $e_0$, we could add some group element to all of the variables in one component so that all previously satisfied edges are still satisfied, and the bundle of $e_0$ is satisfied as well. Therefore, any optimal assignment must be one of the possible assignments output by this algorithm.
	
	Suppose $\widetilde{H}$ has $n$ vertices. Since there are at most $d^n$ spanning trees of $\widetilde{H}$ \cite{NumberOfSpanningTrees} that could be chosen in step \ref{itmNUGLGSoundnessU3PickTree}, and $(2^\ell)^{n - 1} \leq 2^{\ell n}$ functions $z^*$ that could be chosen in step \ref{itmNUGLGSoundnessU3PickVectorsFromBundles}, this algorithm has at most $d^n2^{\ell n}$ computation paths.
	
	As $\widetilde{H}$ has $\frac{nd}{2}$ edges, the expected number of bundles satisfied by any given assignment $x_v$ output by this algorithm is
	$$n - 1 + \left(\frac{nd}{2} - (n - 1)\right)2^{\ell - m},$$
	since the $n - 1$ bundles within $T$ are all satisfied, and each of the other $(\frac{nd}{2} - (n - 1))$ bundles are satisfied with probability $2^{\ell - m}$, independently\footnote{This is the probability that a given $m$-dimensional vector over $\ff_2$ lies in a randomly chosen affine subspace of dimension $\ell$. All that is necessary for this to be true is that the $b(v_1, v_2)$ vectors are chosen randomly; the $Z(v_1, v_2)$ subspaces and the choices made by the algorithm can be arbitrary (as long as they do not depend on $b$). }. Applying Hoeffding's inequality \cite{Hoeffding}, the probability that $x_v$ satisfies more than
	$$n - 1 + \left(\frac{nd}{2} - (n - 1)\right)\left(2^{\ell - m} + \frac{\alpha}{2}\right)$$
	bundles is at most
	$$\exp\left(-2\left(\frac{\alpha}{2}\right)^2\left(\frac{nd}{2} - (n - 1)\right)\right).$$
	By the union bound, the probability that there is \emph{some} computation path giving an assignment satisfying more than this many bundles is at most
	\begin{align*}
		&d^n 2^{\ell n} \cdot \exp\left(-2\left(\frac{\alpha}{2}\right)^2\left(\frac{nd}{2} - (n - 1)\right)\right)\\
		=&\exp(n \ln(d) + \ell n \ln(2)) \cdot \exp\left(-\left(\frac{\alpha^2}{2}\right)\left(\frac{n(d - 2)}{2} + 1\right)\right)\\
		=&\exp\left(n(\ln(d) + \ell\ln(2))-\left(\frac{\alpha^2}{2}\right)\left(\frac{n(d - 2)}{2} + 1\right)\right)\\
		\leq&\exp\left(n(\ln(d) + \ell\ln(2))-\left(\frac{\alpha^2}{2}\right)\left(\frac{n(d - 2)}{2}\right)\right)\\
		=&\exp\left(n\left(\ln(d) + \ell\ln(2)-\frac{\alpha^2(d - 2)}{4}\right)\right)\\
		\leq&\exp\left(n\left(\ln(d) + \ell\ln(2)-\frac{\alpha^2d}{8}\right)\right) \snc{(\ref{equNUGLGParameterDStrictlyGreater}) \implies d \geq 4}\\
		\leq&\exp\left(n\left(\frac{\alpha^2d}{16} + \ell\ln(2) - \frac{\alpha^2d}{8}\right)\right) \stextn{from (\ref{equNUGLGParameterDGeqComplicated}) and the fact that $\ln(\varepsilon) \leq 0$}\\
		=&\exp\left(n\left(\ell\ln(2) - \frac{\alpha^2d}{16}\right)\right)\\
		\leq&\exp\left(n\left(\ell\ln(2) - \left(\ln(d) + 2 + \ln(2) - \ln(\varepsilon)\right)\right)\right) \stext{from (\ref{equNUGLGParameterDGeqComplicated})}\\
		\leq&\exp\left(n\left(\left(\log_2(d) + 2\log_2(e) + 1\right)\ln(2) - \left(\ln(d) + 2 + \ln(2) - \ln(\varepsilon)\right)\right)\right) \stextn{from (\ref{equNUGLGParameterLDef})}\\
		=&\exp\left(n\left(\left(\frac{\ln(d) + 2}{\ln(2)} + 1\right)\ln(2) - \left(\ln(d) + 2 + \ln(2) - \ln(\varepsilon)\right)\right)\right)\\
		=&\exp\left(n\ln(\varepsilon)\right)\\
		=&\varepsilon^{n}\\
		\leq& \varepsilon \bcause{n \geq 1,\ \varepsilon \leq 1}.\\
	\end{align*}
	
	Since all optimal assignments arise from one of these computation paths, it follows that, with probability at least $1 - \varepsilon$, the optimal fraction of bundles which can be satisfied in $\widetilde{U}_2$ is at most
	\begin{align*}
		&\frac{n - 1 + \left(\frac{nd}{2} - (n - 1)\right)\left(2^{\ell - m} + \frac{\alpha}{2}\right)}{\frac{nd}{2}}\\
		<&\frac{n + \left(\frac{nd}{2}\right)\left(2^{\ell - m} + \frac{\alpha}{2}\right)}{\frac{nd}{2}}\\
		=&\frac{2}{d} + 2^{\ell - m} + \frac{\alpha}{2}\\
		\leq&\frac{2}{d} + 2^{\ell - \left(\ell - \log_2\left(\left(\frac12 - \gamma\right)\alpha - \frac{2}{d}\right)\right)} + \frac{\alpha}{2} \stext{from (\ref{equNUGLGParameterMDef})}\\
		=& \frac{2}{d} + \left(\left(\frac12 - \gamma\right)\alpha - \frac{2}{d}\right) + \frac{\alpha}{2}\\
		=&(1 - \gamma)\alpha.
	\end{align*}
	Since each bundle contains $2^\ell$ contradictory constraints, this is a $(1 - \gamma) \left(\frac{\alpha}{2^\ell}\right)$ fraction of the constraints of $\widetilde{U}_2$.\thmspace
\end{proof}

We can now prove Lemma \ref{lemNUGLGSoundnessU2}.\thmspace
\newtheorem*{R4}{Lemma~\ref{lemNUGLGSoundnessU2}}
\begin{R4}
	\lemNUGLGSoundnessUTwo
\end{R4}
\begin{proof}
	By Lemma \ref{lemNUGLGMostEdgesGood}, the probability that less than a $(1 - \gamma)$ fraction of edges of $\widetilde{H}$ are good edges is at most $\frac12$. By Lemma \ref{lemNUGLGSoundnessUTilde}, the probability that $\widetilde{U}_2$ is $(1 - \gamma) \left(\frac{\alpha}{2^\ell}\right)$-satisfiable is at most $\varepsilon$. By the union bound, the probability that either of these two events occurs is at most $\frac12 + \varepsilon$, so the probability that neither event occurs is at least $\frac12 - \varepsilon$. So it suffices to prove that, whenever at least a $(1 - \gamma)$ fraction of the edges of $\widetilde{H}$ are good edges, if $\widetilde{U}_2$ is not $(1 - \gamma) \left(\frac{\alpha}{2^\ell}\right)$-satisfiable, then $U_2$ is not $\left(\frac{\alpha}{2^\ell}\right)$-satisfiable.
	
	We instead prove the contrapositive, that if at least a $\left(\frac{\alpha}{2^\ell}\right)$ fraction of constraints are satisfiable in $U_2$, then at least a $(1 - \gamma)\left(\frac{\alpha}{2^\ell}\right)$ fraction of constraints are satisfiable in $\widetilde{U}_2$. Suppose that $\widetilde{U}_2$ has a total of $c$ constraints. Then $U_2$ has at least $(1 - \gamma)c$ constraints. So if at least a $\left(\frac{\alpha}{2^\ell}\right)$ fraction of constraints are satisfiable in $U_2$, it means that at least $\left(\frac{\alpha}{2^\ell}\right)(1 - \gamma)c$ constraints of $U_2$ are satisfied by some assignment $x_v$. Since $U_2$ and $\widetilde{U}_2$ have the same variable set, and all of the constraints of $U_2$ are also constraints of $\widetilde{U}_2$, it follows that $x_v$ must satisfy $\left(\frac{\alpha}{2^\ell}\right)(1 - \gamma)c$ constraints of $\widetilde{U}_2$ as well, that is, at least a $\left(\frac{\alpha}{2^\ell}\right)(1 - \gamma)$ fraction of constraints.
\end{proof}


\subsection{Proof of Lemma \ref{lemNUGLGRadiusProperty} (paths of length $r$ can be made consistent)}\label{appNUGLGRadiusProperty}

Here we prove Lemma \ref{lemNUGLGRadiusProperty}.\thmspace
\newtheorem*{R5}{Lemma~\ref{lemNUGLGRadiusProperty}}
\begin{R5}
	\lemNUGLGRadiusProperty
\end{R5}
\begin{proof}
	Since $H$ contains only good edges and $p$ has length at least $r$, there exists a set of vectors
	$$B \subseteq \bigcup_{i \in [n]} Z(v_{i - 1}, v_i)$$
	forming a basis of $\fm$. Write $h(i)$ for the number of basis vectors in $Z(v_{i - 1}, v_i)$, and denote these vectors by
	$$B = \bigcup_{i \in n} \{z_{i,j} \suchthat j \in [h(i)]\},$$
	where each $z_{i,j} \in Z(v_{i - 1}, v_i)$. Since $B$ is a basis, there exist coefficients $c_{i, j}$ such that
	\begin{equation}\label{equGStarBasis}
		g^*(v_0) - g^*(v_n) - \sum_{i \in [n]} b(v_{i - 1}, v_{i}) = \sum_{i \in [n]} \sum_{j \in [h(i)]} c_{i, j}z_{i, j}.
	\end{equation}
	For each $i$ in order from $1$ to $n$, inductively define
	$$g^*(v_i) := g^*(v_{i - 1}) - \sum_{j \in [h(i)]} c_{i, j}z_{i, j} - b(v_{i - 1}, v_i).$$
	Note that, by expanding the inductive definition for $g^*(v_n)$, we have
	\begin{align*}
		g^*(v_n) &= g^*(v_{n - 1}) - \sum_{j \in [h(n)]} c_{n, j}z_{n, j} - b(v_{n - 1}, v_n)\\
		&= g^*(v_{n - 2}) - \sum_{j \in [h(n - 1)]} c_{n - 1, j}z_{n - 1, j} - b(v_{n - 2}, v_{n - 1}) \push\hspace{.64cm} - \sum_{j \in [h(n)]} c_{n, j}z_{n, j} - b(v_{n - 1}, v_n)\\
		&= g^*(v_{n - 3}) - \sum_{j \in [h(n - 2)]} c_{n - 2, j}z_{n - 2, j} - b(v_{n - 3}, v_{n - 2}) \push\hspace{.64cm} - \sum_{j \in [h(n - 1)]} c_{n - 1, j}z_{n - 1, j} - b(v_{n - 2}, v_{n - 1}) \push\hspace{.64cm} - \sum_{j \in [h(n)]} c_{n, j}z_{n, j} - b(v_{n - 1}, v_n)\\
		&= \dots\\
		&= g^*(v_0) - \sum_{i \in [n]} \left(b(v_{i - 1}, v_{i}) + \sum_{j \in [h(i)]} c_{i, j}z_{i, j}\right)\\
	\end{align*}
	so our inductive definition agrees with the original definition by (\ref{equGStarBasis}). For any $i \in [n]$ and any arbitrary elements $g_{i - 1}, g_i, z \in \fm$,
	\begin{align*}
		&& x_{v_{i}}^{g_{i}} - x_{v_{i - 1}}^{g_{i - 1}} = z &\txt{ is an equation in } \mathcal{G}(U_1)\\
		\iff&& (x_{v_{i}}^{g_{i}} + g_{i}) - (x_{v_{i - 1}}^{g_{i - 1}} + g_{i - 1}) = z &\txt{ is an equation in } U_1\\
		\iff&& (x_{v_{i}}^{g_{i}} + g_{i} + g^*(v_{i - 1})) \hspace{.76cm}\\&& \txt{} - (x_{v_{i - 1}}^{g_{i - 1}} + g_{i - 1} + g^*(v_{i - 1})) = z &\txt{ is an equation in } U_1\\
		\iff&& (x_{v_{i}}^{g_{i}} + g_{i} + g^*(v_{i - 1})) \\&& \txt{} - (x_{v_{i - 1}}^{g_{i - 1}} + g_{i - 1} + g^*(v_{i - 1})) \\&& = z + \sum_{j \in [h(i)]} c_{i, j}z_{i, j} &\txt{ is an equation in } U_1 \\&&\snc{\sum_{j \in [h(i)]} c_{i, j}z_{i, j} \in Z(v_{i - 1}, v_{i})}\\
		\iff&& (x_{v_{i}}^{g_{i}} + g_{i} + g^*(v_{i - 1})) \\&& \txt{} - (x_{v_{i - 1}}^{g_{i - 1}} + g_{i - 1} + g^*(v_{i - 1})) \\&& = z + \sum_{j \in [h(i)]} c_{i, j}z_{i, j} + b(v_{i - 1}, v_{i}) &\txt{ is an equation in } U_2\\
		\iff&& (x_{v_{i}} + g_{i} + g^*(v_{i - 1}) - \sum_{j \in [h(i)]} c_{i, j}z_{i, j} \hspace{.76cm}\\&& \txt{} - b(v_{i - 1}, v_i)) - (x_{v_{i - 1}} + g_{i - 1} + g^*(v_{i - 1})) = z &\txt{ is an equation in } U_2\\
		\iff&& x_{v_{i}}^{g_{i} + g^*(v_{i - 1}) - \sum_{j \in [h(i)]} c_{i, j}z_{i, j} - b(v_{i - 1}, v_i)} \hspace{.76cm}\\&& \txt{} - x_{v_{i - 1}}^{g_{i - 1} + g^*(v_{i - 1})} = z &\txt{ is an equation in } \mathcal{G}(U_2)\\
		\iff&& x_{v_{i}}^{g_{i} + g^*(v_{i})} - x_{v_{i - 1}}^{g_{i - 1} + g^*(v_{i - 1})} = z &\txt{ is an equation in } \mathcal{G}(U_2)\\
		\iff&& f(x_{v_{i}}^{g_{i}}) - f(x_{v_{i - 1}}^{g_{i - 1}}) = z &\txt{ is an equation in } \mathcal{G}(U_2),
	\end{align*}
	so $f$ is a partial isomorphism.
\end{proof}


\subsection{Proof of Lemma \ref{lemGirth} (no paths in $F_i(u)$ with endpoints in $T_{i - 1}$)}\label{appGirth}

Here we prove Lemma \ref{lemGirth}. First, we need the following result about $T_i(u)$.\thmspace
\begin{lemma}\label{lemTreePath}
	On any round $i$, for any vertex $u \in V(H)$, any path in $T_i(u)$ passes through at most $k$ vertices in $P_i(u)$.
\end{lemma}
\begin{proof}
	Let $p = v_0, \seq{v}{n}$ be a path in $T_i(u)$. Consider the following map $h: p \cap P_i(u) \to V(T_i(u))$:
	$$h(v_i) := \threecases{\txt{if } \deg(v_i) < 3}{v_i}{}{\txt{some pebbled vertex (or $u$) reachable}}{\txt{if } \deg(v_i) \geq 3}{\txt{from $v_i$ in $T_i(u) \setminus \{v_{i - 1}, v_{i + 1}\}$}}$$
	Note that such a vertex always exists when $v_i$ has degree at least 3, and is necessarily different from all other vertices in the image of $h$. Thus, $h$ is injective. Also, since vertices in $P_i(u)$ of degree less than 3 must be pebbled, the output of $h(v_i)$ must always be a pebbled vertex. Thus, we have an injection from $p \cap P_i(u)$ to a set of pebbled vertices (plus $u$), of which there are at most $k$ (since one pebble pair has been picked up), so $\abs{p \cap P_i(u)} \leq k$.
\end{proof}

Now we can prove Lemma \ref{lemGirth}.\thmspace
\newtheorem*{R6}{Lemma~\ref{lemGirth}}
\begin{R6}
	\lemGirth
\end{R6}
\begin{proof}
	Suppose toward a contradiction that there was such a path $p_1$, joining $v_1, v_3 \in T_{i - 1}$. Let $v_2$ be the first vertex along the path $p_1$ which is contained $T_{i - 1}$, excluding $v_1$ (it could just be $v_3$ if there are no earlier places where $p$ crosses $T_{i - 1}$). Since $T_{i - 1}$ is connected, there must be some path $p_2$ joining $v_1$ and $v_2$ in $T_{i - 1}$. Since $p_1$ is contained in $F_i(u)$, which shares no edges with $T_{i - 1}$, $p_1$ and $p_2$ share no edges. Aside from $v_1$ and $v_2$, they do not share any common vertices either, from the way that $v_2$ was chosen. So together, $p_1$ and $p_2$ form a cycle. Since $p_1$ is contained within $T_i(u)$, by Lemma \ref{lemTreePath} it intersects at most $k$ vertices in $P_i(u)$. Since, additionally, $p_1$ is contained in $F_i(u)$, the length of each of the $\leq(k + 1)$ segments between vertices in $P_i(u)$ and the endpoints is strictly less than $r$. Thus, $p_1$ has length strictly less than $(k + 1)r$. Since $p_2$ is contained within $T_{i - 1}$, which is minimal, $p_2$ cannot contain any subpaths of length $(k + 1)r$ which do not intersect $P_{i - 1}$, for otherwise, swapping out such a subpath for $p_1$ would yield a strictly smaller tree. Applying Lemma \ref{lemTreePath} to round $i - 1$ and vertex $u^*_{i - 1}$, we have that at most $k$ vertices of $p_2$ intersect $P_{i - 1}$, so $p_2$ has length at most $k (k + 1) r$. Thus, concatenating $p_1$ and $p_2$ yields a cycle of size strictly less than
	$$(k + 1) r + k (k + 1) r = (k + 1)^2 r$$
	in $H$. This contradicts the fact that $H$ was chosen to have girth at least $(k + 1)^2 r$. Hence, no such path $p_1$ can exist.
\end{proof}


\subsection{Derivation of bound on the growth of $q$}\label{appBigO}

Here we explicitly derive the bound $q = O(\frac{1}{\alpha^2} \log(\frac{1}{\alpha}))$ from Theorem \ref{thmUGLowGapMain}. First, we need the following lemma.\thmspace
\begin{lemma}\label{lemBigO}
	For any function $f: \zz_{\geq 1} \to \zz_{\geq 1}$,
	$$f(n) = O(n^2 \log(f(n))) \implies f(n) = O(n^2 \log (n)).$$
\end{lemma}
\begin{proof}
	Suppose that there exist $c_0, n_0$ such that, for all $n \geq n_0$,
	$$f(n) \leq c_0 n^2 \log_2(f(n)).$$
	Since $\log_2(m) \leq 2\sqrt{m}$ for all positive integers $m$, it follows that, when $n \geq n_0$,
	\begin{align*}
		f(n) \leq c_0 n^2 \log_2(f(n)) \leq c_0 n^2 \left(2\sqrt{f(n)}\right) &\implies \sqrt{f(n)} \leq 2 c_0 n^2\\
		&\implies f(n) \leq 4c_0 n^4.
	\end{align*}
	Let $c_1 := 5c_0$ and let $n_1 := \max\{n_0, 4c_0\}$. Then, for all $n \geq n_1$,
	\begin{align*}
		f(n) &\leq c_0 n^2 \log_2(f(n))\\
		&\leq c_0 n^2 \log_2(4 c_0 n^4)\\
		&= c_0 n^2 \left(\log_2(4c_0) + 4\log_2(n)\right)\\
		&\leq c_0 n^2 \left(\log_2(n) + 4\log_2(n)\right) \snc{n \geq n_1 \geq 4c_0}\\
		&= c_1 n^2 \log_2(n),
	\end{align*}
	so $f(n) = O(n^2 \log(n))$.
\end{proof}

Let $f(n)$ denote the smallest integer $d$ satisfying (\ref{equNUGLGParameterDGeqComplicated}) for $\alpha = \frac{1}{n}$ (note that we don't have to worry about (\ref{equNUGLGParameterDStrictlyGreater}) since $\gamma = \frac14$). By (\ref{equNUGLGParameterDGeqComplicated}),
$$f(n) = O\left(\left(\frac{1}{\alpha}\right)^2 \log (f(n))\right) = O(n^2 \log (f(n))),$$
so by Lemma \ref{lemBigO}, $f(n)= O(n^2 \log(n))$. In other words, we can take $d = O(\frac{1}{\alpha^2} \log(\frac{1}{\alpha}))$ to satisfy (\ref{equNUGLGParameterDGeqComplicated}). Then it is not hard to see that (\ref{equNUGLGParameterLDef}) and (\ref{equNUGLGParameterMDef}) imply $q = 2^m = O(\frac{1}{\alpha^2} \log(\frac{1}{\alpha}))$ as well.

\end{document}